\documentclass[preprint,12pt]{elsarticle}

\usepackage{paralist}
\usepackage{verbatim}
\usepackage{url}
\usepackage{epsfig}

\usepackage{algorithm}
\usepackage[noend]{algpseudocode}

\usepackage{amsmath}
\usepackage{amssymb}
\usepackage{wasysym}
\usepackage{amsthm}
\usepackage{multirow}

\input{psfig}

\usepackage{graphicx}
\DeclareGraphicsExtensions{.pdf,.png,.jpg}
\usepackage{epsfig}
\usepackage{subfigure} 

\newtheorem{lemma}{Lemma}
\newtheorem{corollary}[lemma]{Corollary}

\journal{Parallel Computing: SI: Scientific graph analysis}

\begin{document}

\begin{frontmatter}



\title{Parallel Heuristics for Scalable Community Detection}


\author{
Hao Lu 
%\fnref{labelWSU}
}
\ead{luhowardmark@wsu.edu}
%\fntext[labelWSU]{WSU}
\address{School of Electrical Engineering and Computer Science, 
Washington State University, Pullman WA 99164
%\fnref{labelWSU}
}

\author{
Mahantesh Halappanavar 
%\fnref{labelPNNL}
}
\ead{hala@pnnl.gov}
\address{
Computational Sciences and Mathematics Division, Pacific Northwest National Laboratory, Richland, WA 99354
%\fnref{labelPNNL}
}

\author{
Ananth Kalyanaraman\corref{cor1}
%\fnref{labelWSU}
}
\ead{ananth@eecs.wsu.edu}
%\ead[url]{http://www.eecs.wsu.edu/~ananth}
%\fntext[labelWSU]{WSU}
\cortext[cor1]{Corresponding author}
\address{School of Electrical Engineering and Computer Science, 
Washington State University, Pullman WA 99164
%\fnref{labelWSU}
}

\begin{abstract}
Community detection has become a fundamental operation in numerous graph-theoretic applications. It is used to reveal natural divisions that exist within real world networks without imposing prior size or cardinality constraints on the set of communities. 
Despite its potential for application, there is only limited support for community detection on large-scale parallel computers, largely owing to the irregular and  inherently sequential nature of the underlying heuristics. 
In this paper, we present parallelization heuristics for fast community detection using the {\it Louvain} method as the serial template. The Louvain method is an iterative heuristic for  modularity optimization. Originally developed by Blondel {\it et al.} in 2008, the method has become increasingly popular owing to its ability to detect high modularity community partitions in a fast and memory-efficient manner. However, the method is also inherently sequential, thereby limiting its scalability. Here, we observe certain key properties of this method that present challenges for its parallelization, and consequently propose heuristics that are designed to break the sequential barrier. For evaluation purposes, we implemented our heuristics using OpenMP multithreading, and tested them over real world graphs derived from multiple application domains (e.g., internet, citation, biological). 
Compared to the serial Louvain implementation, our parallel implementation is able to produce community outputs with a higher modularity for most of the inputs tested, in comparable number or fewer iterations, while providing absolute speedups of up to $16\times$ using 32 threads. 
%In addition, our parallel implementation was able to exhibit weak scaling properties on up to 32 threads. 
\end{abstract}


\end{frontmatter}



\section{Introduction}
\label{secIntro}

Community detection, or graph clustering, is becoming pervasive in the data analytics of various fields including (but not limited to) scientific computing, life sciences, social network analysis, and internet applications \cite{fortunato_community_2010}. As data grows at explosive rates, the need for scalable tools to support fast implementations of complex network analytical functions such as community detection is critical. Given a graph, the problem of community detection is to compute a partitioning of vertices into communities that are closely related within and weakly across communities. Modularity is a metric that can be used to measure the quality of communities detected \cite{newman_finding_2004}.  Modularity maximization is an NP-Complete problem \cite{brandes_modularity_2008} and therefore fast approximation heuristics are used in practice. One such heuristic is the Louvain method~\cite{blondel_fast_2008}.
 
Our basis for selecting the Louvain heuristic for parallelization hinges on its increasing popularity within the user community and owing to its strengths in algorithmic and qualitative robustness. With well over 1,700 citations to the original paper (as of this writing), the user base for this method has been rapidly expanding in the last few years. %Yet, there is no scalable parallel implementation available for this heuristic. 
As network sizes continue to grow rapidly into scales of tens or even hundreds of billions of edges \cite{dimacs10_10th_????}, the memory and runtime limits of the serial implementation are likely to be tested. However, parallelization of this inherently serial algorithm can be challenging (as discussed in Section~\ref{secChallenges}). 

The parallel solutions presented in this paper (Section~\ref{secMethods}) provide a way to overcome key scalability challenges. In devising our algorithm, we factored in the need to parallelize without compromising the quality of the original serial heuristic and yet be capable of achieving substantial speedup. Where possible, we also factored in the need for guaranteeing stability in output across different platforms and programming models. The resulting algorithm, presented in Section~\ref{secParallelAlgo}, is a combination of heuristics that can be implemented on both shared and distributed memory machines.
As demonstrated in our experimental section (Section~\ref{secExp}), our multi-threaded implementations output results that have either a higher or comparable modularity to that of the serial method, and is able to reduce the time to solution by factors of up to $16\times$. These observations are supported over a number of real-world networks. 

{\bf Contributions:} The main contributions of this paper are:
\begin{enumerate}[i)]
\item
 Introduction of novel and effective heuristics for parallelization of the Louvain algorithm on multithreaded architectures; 
\item
 Experimental studies using 11 real-world networks obtained from varied sources including the DIMACS10 challenge website, University of Florida sparse matrix collection and biological databases; and
%\item 
%A simple and clarified derivation of modularity computation that can benefit other researchers exploring community detection algorithms; and
\item 
A thorough comparative study of the performance and related trade-offs among the different parallel heuristics along with the serial Louvain method. 
%Demonstration of the effectiveness of our parallel heuristics through rigorous comparison with the respective serial solutions.   
\end{enumerate}

\vspace*{-0.1in}
\section{Problem statement and notation}
\label{secProblem}

Let $G(V,E,\omega)$ be an undirected weighted graph, where $V$ is the set of vertices, $E$ is the set of edges and $\omega(.)$ is a weighting function that maps every edge in $E$ to a non-zero, positive weight\footnote{If the graph is unweighted, then we treat every edge to be of weight 1.}. 
In the input graph,  edges that connect a vertex to itself are allowed --- i.e., $(i,i)$ can be a valid edge. However, multi-edges are not allowed. 
Let the adjacency list of $i$ be denoted by $\Gamma(i) = \{j | (i,j)\in E\}$. 
Let $k_i$ denote the weighted degree of vertex $i$ --- i.e., 
$k_i = \sum_{\forall j\in\Gamma(i)}\omega(i,j)$.
We will use $n$ to denote the number of vertices in $G$; $M$ to denote the \emph{number} of edges in the graph; and $m$ to denote the sum of all edge weights  --- i.e., $m=\frac{1}{2}\sum_{\forall i\in V}k_i$.  

A \emph{community} within graph $G$ represents a (possibly empty\footnote{The notion of empty communities does not have practical relevance. We have intentionally defined it this way so as to make our later algorithmic descriptions easier. It is guaranteed, however, that all output communities at the end of our algorithm will be non-empty subsets.}) subset of  $V$. In practice, for community detection, we are interested in partitioning the vertex set $V$ into an arbitrary number of \emph{disjoint} non-empty communities, each with an arbitrary size ($>0$ and $\leq n$).
We call a community with just one element as a \emph{singlet} community. 
%Let $k$ denote the number of communities in a given partition of $V$. 
We will use $C(i)$ to denote the community that contains vertex $i$ in a given partitioning of $V$.  
We use the term \emph{intra-community edge} to refer to an edge that connects two vertices of the same community. All other edges are referred to as \emph{inter-community edges}.
Let $E_{i\rightarrow C}$ refer to the set of all edges connecting vertex $i$ to vertices in community $C$. 
%Note that if $C=C(i)$, then all the edges in $E_{i\rightarrow C}$ are intra-community edges; and inter-community edges otherwise.  
And let $e_{i\rightarrow C}$ denote the sum of the edge weights for the edges in $E_{i\rightarrow C}$. 
\begin{eqnarray}
  \label{eqnEIC}
  e_{i\rightarrow C} &=& \sum_{\forall (i,j)\in E_{i\rightarrow C}} \omega(i,j)
\end{eqnarray}
Let $a_C$ denote the sum of the degrees of all the vertices in community $C$  (also referred to as \emph{community degree}).  
\begin{eqnarray}
  \label{eqnAC}
a_c &=& \sum_{\forall i\in C}k_i  
\end{eqnarray}

%Note that in the above definition, each intra-community edge in community $C$ is double-counted. For convenience, we will sometimes refer to this term as the \emph{degree of a community}.

{\bf Modularity:} Let $P=\{C_1,C_2,\ldots C_k\}$ denote the set of all communities in a given partitioning of the vertex set $V$ in $G(V,E,\omega)$, where $1\leq k\leq n$. Consequently, the \emph{modularity} (denoted by $Q$) of the partitioning $P$ is given by the following expression \cite{newman_finding_2004}:
\begin{eqnarray}
  \label{eqnQ}
  Q \; = \; \frac{1}{2m} \sum_{\forall i\in V}e_{i\rightarrow C(i)} \; - \;  \sum_{\forall C\in P} (\frac{a_C}{2m}\cdot \frac{a_C}{2m})
\end{eqnarray}

%Intuitively, modularity is a statistical measure for assessing the quality of a given community-wise partitioning (or equivalently, ``clustering''). A ``good'' clustering method is one that  clusters closely related elements (vertices) as part of the same community (or  ``cluster'') while separating weakly related elements into different clusters. In other words, the goal becomes one of maximizing  intra-community links while keeping the number  of inter-community edges low. This explains the first term in Eqn.~(\ref{eqnQ}).  However, if the goal is simply to maximize the contribution from intra-community edges, then one could potentially assign all vertices into one community. But such a solution is likely to be meaningless in practice.  To overcome this problem, the second term in the Eqn.~(\ref{eqnQ}) was introduced. This term represents the fraction of intra-community edges one would expect in an ``equivalent'' graph (i.e., another graph with the same numbers of vertices and edges, and the same vertex degrees) but with just the edges randomly reconnected.

%Note that the term $\frac{a_C}{2m}\times \frac{a_C}{2m}$ represents the product of the probabilities of an edge originating at a vertex of a given community $C$ and also ending at a vertex of the same community in such a graph.

%Note from Eqn.~(\ref{eqnQ}) the theoretical upperbound for modularity is 1 for any input graph. However, the expression can also evaluate to be negative. 

Modularity is not an ideal metric for community detection and issues such as resolution limit have been identified \cite{fortunato_community_2010,traag_narrow_2011}; a few variants of modularity definitions have also been devised \cite{traag_narrow_2011,bader_modularity_2009,berry_tolerating_2011}. 
However, the definition provided in Eqn.~(\ref{eqnQ}) continues to be the more widely adopted version in practice, including in the Louvain method \cite{blondel_fast_2008}, and therefore, we will use that definition for this paper. 

{\bf Community detection:} Given $G(V,E,\omega)$, the problem of community detection is to compute a partitioning $P$ of communities that maximizes modularity.

This problem has been shown to be NP-Complete \cite{brandes_modularity_2008}. 
Note that this problem is different from graph partitioning problem and its variants \cite{hendrickson_graph_2000}, where the number of clusters and the rough size distribution of those target clusters are known {\it a priori}. In the case of community detection, both quantities are unknown prior to computation. In fact they encapsulate the input properties that one seeks  to discover out of the community detection exercise.

\section{The Louvain algorithm}

In 2008, Blondel {\it et al.} presented an algorithm for the community detection \cite{blondel_fast_2008}. The method, called the \emph{Louvain} method, is an iterative, greedy heuristic capable of producing a hierarchy of communities. The main idea of the algorithm is rather simple and can be summarized as follows: Starting with each vertex in a separate community, the algorithm progresses from one iteration to another until the modularity \emph{gain} becomes negligible (as defined by a predefined threshold). Within each \emph{iteration}, the algorithm linearly scans the vertices in an arbitrary but predefined order. For every vertex $i$, all its neighboring communities (i.e., the communities containing $i's$ neighbors) are examined and the modularity gain that would result if $i$ were to move to each of those neighboring communities from its current community is calculated.  Once the gains are calculated, the algorithm assigns a neighboring community that would yield the maximum modularity gain, as the new community for $i$ (i.e., new $C(i)$), and updates the corresponding data structures that it maintains for the source and target communities. Alternatively, if all gains turn out to be negative, the vertex stays in its current community. An iteration ends once all vertices are linearly scanned in this fashion. After a phase terminates, the algorithm proceeds to the next \emph{phase} by collapsing all vertices of a community to a single ``meta-vertex''; placing an edge from that meta-vertex to itself with an edge weight that is the sum of weights of all the intra-community edges within that community; and placing an edge between two meta-vertices with a weight that is equal to the sum of the weights of all the inter-community edges between the corresponding two communities. The result is a condensed graph $G^\prime(V^\prime,E^\prime,\omega^\prime)$, which then becomes the input to the next phase. Subsequently, multiple phases are carried out until the modularity score converges. Note that each phase represents a coarser level of hierarchy in the community detection process.

%Let $C(i,k)$ denote the new community assignment made for vertex $i$ during iteration $k$ (of any phase) of the algorithm, for some $k\geq 1$.  Initially, $C(i,0)=\{i\}$, $\forall i\in V$\footnote{For ease of exposition and without loss of generality, we will present the below details assuming that the algorithm is in its first phase (hence $i\in V$).} .  The value for $C(i,k)$ will be determined using the community assignments 

At any given iteration, let $\Delta Q_{i\rightarrow C(j)}$ denote the modularity gain that would result from moving a vertex $i$ from its current community $C(i)$ to a different community $C(j)$.  This term is given by:
\begin{eqnarray}
  \label{eqnDeltaQ}
  \Delta Q_{i\rightarrow C(j)}  &=&   \frac{e_{i\rightarrow C(j)} - e_{i\rightarrow C(i)\setminus \{i\}}}{m} \nonumber \\
   && + \frac{2\cdot k_i\cdot a_{C(i)\setminus \{i\}} - 2\cdot k_i\cdot a_{C(j)} }{(2m)^2}
\end{eqnarray}
Consequently, the new community assignment for $i$ at an iteration is determined as follows:
\begin{eqnarray}
  \label{eqnCommTarget}
  C(i) &=& \arg\max_{C(j)} \Delta Q_{i\rightarrow C(j)}, \forall j\in \Gamma(i)\cup \{i\} 
\end{eqnarray}
%Note that the new $C(i)$ will equal the old $C(i)$ if none of the modularity gains to any of the other neighboring communities evaluates to a positive value. 
In the implementation \cite{louvain_findcommunities_????}, several data structures are maintained such that each instance of $\Delta Q_{i\rightarrow C(j)}$ can be computed in $O(1)$ time. Consequently, the algorithm's time complexity \emph{per} iteration is $O(M)$. While no upper bound has been established on the number of iterations or on the number of phases, it should be evident that the algorithm is guaranteed to terminate with the use of a cutoff for the modularity gain (because of the modularity being a monotonically increasing function until termination). In practice, the method needs only tens of iterations and fewer phases to terminate on most real world inputs.  

%With regards to space complexity, the memory usage within a phase is linearly proportional to the phase's input graph. 
%Note that the output of a phase is simply a vertex to community mapping which is $O(n)$. 
%Even though the algorithm needs to store the sequence of graphs generated at every phase,  the peak memory usage is expected during the first phase. Consequently, the memory complexity is $O(m+n)$. 

\vspace*{-0.1in}
\section{Challenges in parallelization}
\label{secChallenges}

Any attempt at parallelizing the Louvain method should factor in the sequential nature in which the vertices are visited within each iteration and the impact it has on convergence. Visiting the vertices sequentially gives the advantage of working with the latest information available from all the preceding vertices in this greedy procedure. 
Furthermore, in the serial algorithm, when a vertex computes its new community assignment (using Eqn.(\ref{eqnCommTarget})), it does so with the guarantee that no other part of the community structure is concurrently being altered. These guarantees may \emph{not} hold in \emph{parallel}. In other words, if communities are updated in parallel, it could lead to some interesting situations with an impact on the convergence process as described below.

%, let us number the vertices $1\ldots n$ in the order in which they are visited by the serial algorithm within each iteration.

\vspace*{-0.1in}
\subsection{Negative gain scenario}
\label{secNegativeGain}

\begin{figure}[!t]
\centering
\includegraphics[width=0.65\textwidth]{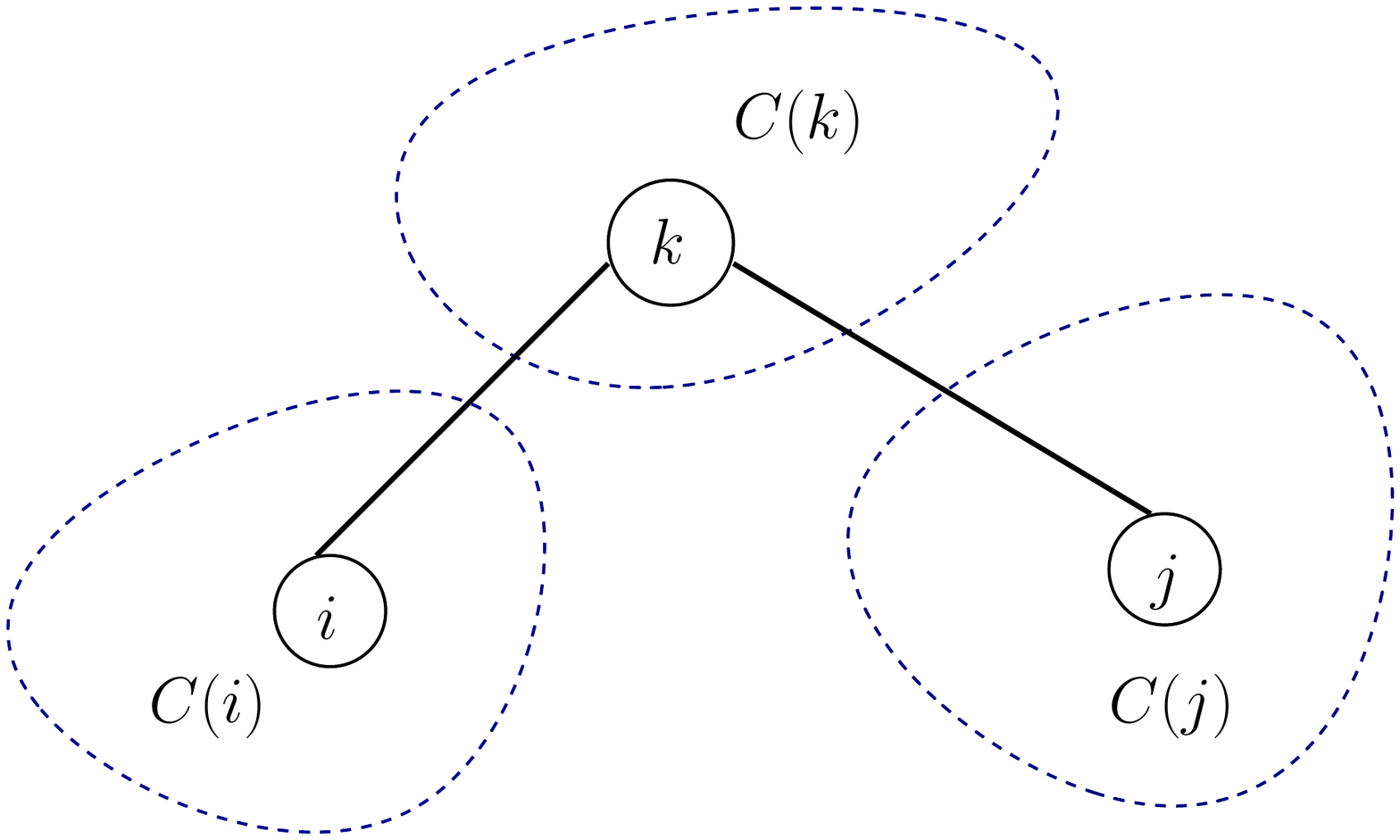}
\caption{Illustration of the negative gain scenario using an example of three vertices (Lemma~\ref{lemma1}).}
\label{figLemma1}
\end{figure}

To illustrate the case in point, consider the example scenario illustrated in Figure~\ref{figLemma1}, where two vertices $i$ and $j$ are both connected to a third vertex $k$ with all three of them in different communities initially --- i.e., $i\in C(i)$, $j\in C(j)$, $k\in C(k)$ s.t. $C(i)\ne C(j) \ne C(k)$.  If both vertices $i$ and $j$ evaluate the possibility of moving to $C(k)$ independently, using Eqn.(\ref{eqnDeltaQ}), then from each of their perspectives, their \emph{predicted} value for the new modularity  is $Q_{old}+\Delta Q_{i\rightarrow C(k)}$ and $Q_{old}+\Delta Q_{j\rightarrow C(k)}$, respectively. However, if both $i$ and $j$ decide to move to $C(k)$ in parallel, then the \emph{actual} value for the new modularity will be $Q_{old}+\Delta Q_{\{i,j\}\rightarrow C(k)}$, where:
\begin{eqnarray}
  \label{eqnLemma1a}
  \Delta Q_{\{i,j\}\rightarrow C(k)} &=& \Delta Q_{i\rightarrow C(k)} + \Delta Q_{j\rightarrow C(k)} \nonumber \\
      &&  + \frac{\omega(i,j)}{m} - \frac{ 2\cdot k_i \cdot k_j }{(2m)^2}
\end{eqnarray}
If $(i,j)\notin E$, $\omega(i,j)=0$, implying:
\begin{eqnarray}\label{eqnLemma1b}
  \Delta Q_{\{i,j\}\rightarrow C(k)} &=& \Delta Q_{i\rightarrow C(k)} + \Delta Q_{j\rightarrow C(k)} \nonumber \\
   && - \frac{ 2\cdot k_i \cdot k_j }{(2m)^2} \nonumber \\
   &\leq &  \Delta Q_{i\rightarrow C(k)} + \Delta Q_{j\rightarrow C(k)}
\end{eqnarray}
Furthermore, if $\Delta Q_{i\rightarrow C(k)} + \Delta Q_{j\rightarrow C(k)} <\frac{ 2\cdot k_i \cdot k_j }{(2m)^2} $
\begin{eqnarray}\label{eqnLemma1c}
  &\Rightarrow &  \Delta Q_{\{i,j\}\rightarrow C(k)} < 0
\end{eqnarray}
On the other hand, if $\frac{\omega(i,j)}{m}>\frac{2\cdot k_i\cdot k_j}{(2m)^2}$ (can be true only if $(i,j)\in E$), then: 
\begin{eqnarray}\label{eqnLemma1d}
  \Delta Q_{\{i,j\}\rightarrow C(k)} &>& \Delta Q_{i\rightarrow C(k)} + \Delta Q_{j\rightarrow C(k)}  
\end{eqnarray}
This is because $\Delta Q_{i\rightarrow C(k)}>0$ and $\Delta Q_{j\rightarrow C(k)}>0$; the latter two inequalities follow from the fact that $i$ and $j$ chose to move to $C(k)$. Note that if this happens, then parallel version could potentially surpass the serial version toward modularity convergence.

%Inequalities (\ref{eqnLemma1b}-\ref{eqnLemma1d}) imply the following lemma:

\begin{lemma}
  \label{lemma1}
  At any given iteration of the Louvain algorithm, if community updates for vertices are performed in parallel, then the net modularity gain achieved cannot be guaranteed to be always positive. 
  \end{lemma}
\begin{proof}
Follows directly from inequality~(\ref{eqnLemma1b}).
\end{proof}

The above lemma has a direct implication on the convergence property of the Louvain method, one way or another. Pessimistically speaking, if the net modularity gain can become negative between consecutive iterations of the algorithm, then there is no theoretical guarantee that the algorithm will terminate. Even if the chances of non-termination turn out to be bleak, it could potentially slow down the rate at which the algorithm progresses toward a solution, causing more number of iterations.  For this reason, the \emph{number of iterations} that the algorithm takes to converge toward the solution and the \emph{quality of the solution} relative to the serial algorithm's can be good indicators of the effectiveness of a parallel strategy. Note that the above example with three vertices can be extended to scenarios where multiple unrelated vertices are trying to enter a community at its periphery without mutual knowledge. 

\vspace*{-0.1in}

\subsection{Swap and local maxima scenarios}
\label{secSwap} 

There exists another scenario that could impede the progression of the parallel algorithm toward a solution. Consider a simple example where two vertices $i$ and $j$ connected by an edge $(i,j)\in E$ s.t., $C(i)=\{i\}$ and $C(j)=\{j\}$. In the interest of increasing modularity, if the two vertices make a decision to move to each other's community concurrently, then such an update could potentially result in both vertices simply swapping their community assignments without achieving any modularity gain. This could also happen in a more generalized setting, where subsets of vertices between two different communities swap their community assignments, each unaware of the other's intent to also migrate. 
%Swap conditions or cyclic migration patterns could lead to a deadlock which need to be detected and surpassed in parallel. 
%Note that this is not a problem with the serial algorithm, as only one decision is made at any given point of time. 

A parallel algorithm also runs the risk of settling on locally optimal decisions. This could happen even in serial;  in parallel such scenarios may arise if a single community gets partitioned into equally weighted sub-communities, in which there is no incentive for any individual vertex to merge with any of the other sub-communities; and yet, if all vertices from each of the sub-communities were to merge together to form a single community the net modularity gain could be positive. An example of this case will be shown later in Section~\ref{secMinDegree}. Getting stuck in a locally optimal solution, however, can be resolved when the algorithm progresses to subsequent phases. 

\vspace*{-0.1in}
\section{Parallel heuristics}
\label{secMethods}

In this section, we present our ideas to tackle the challenges outlined above in parallelizing the Louvain heuristic community detection.

\vspace*{-0.1in}
\subsection{The minimum label heuristic}
\label{secMinDegree}

\begin{figure*}[!t]
\centering
\subfigure{\includegraphics[width=0.31\textwidth]{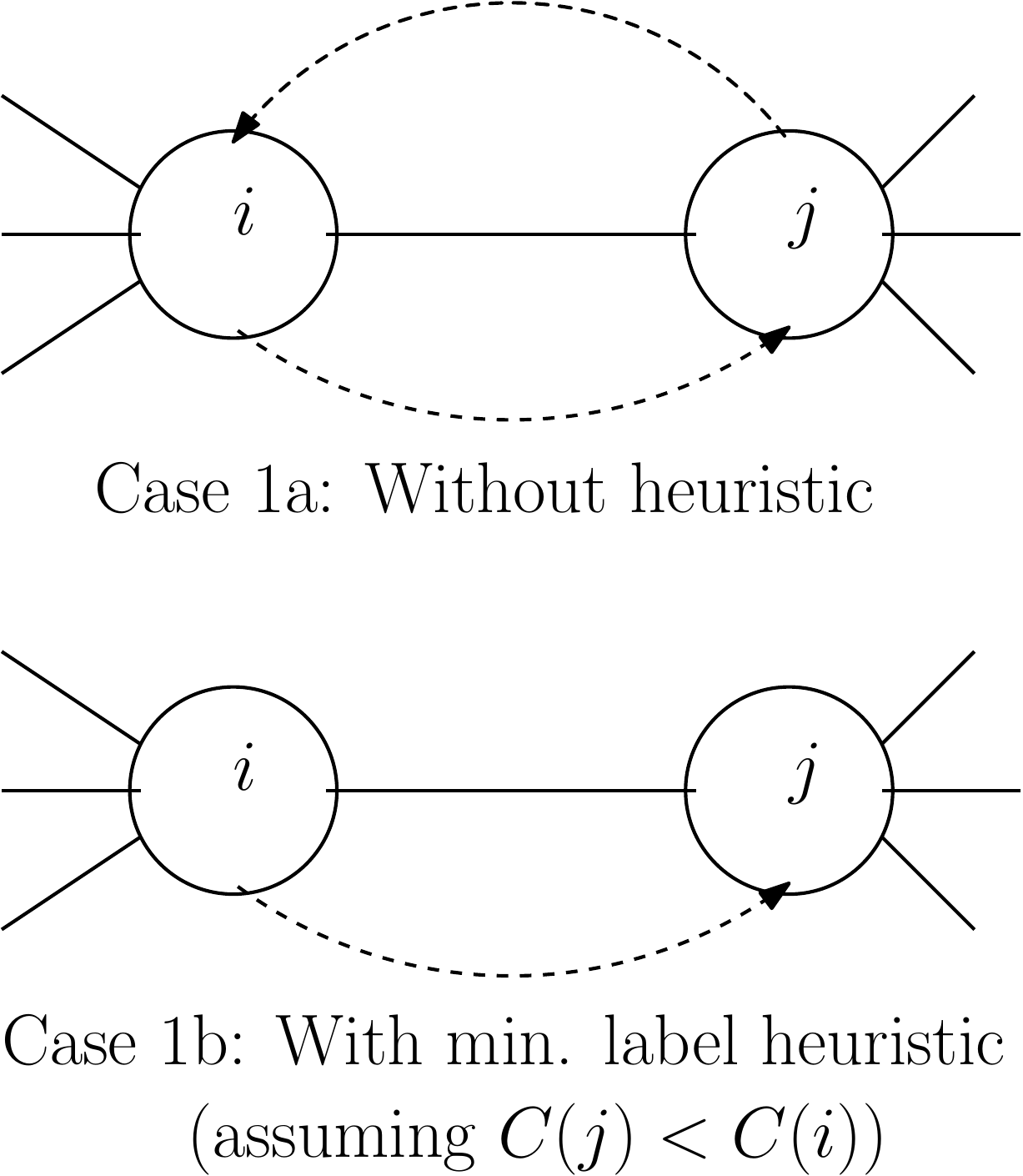} }
\subfigure{\includegraphics[width=0.31\textwidth]{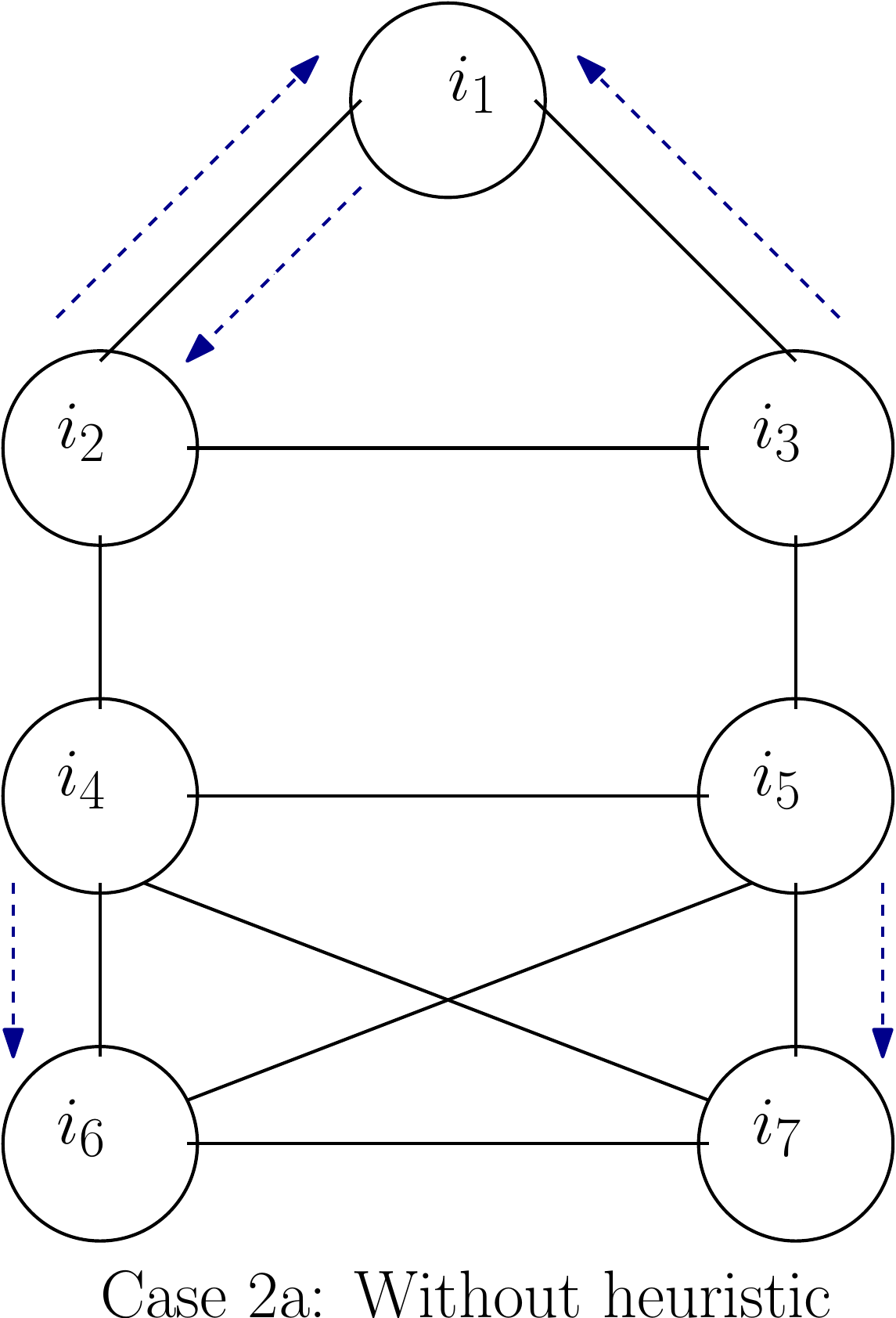} }
\subfigure{\includegraphics[width=0.31\textwidth]{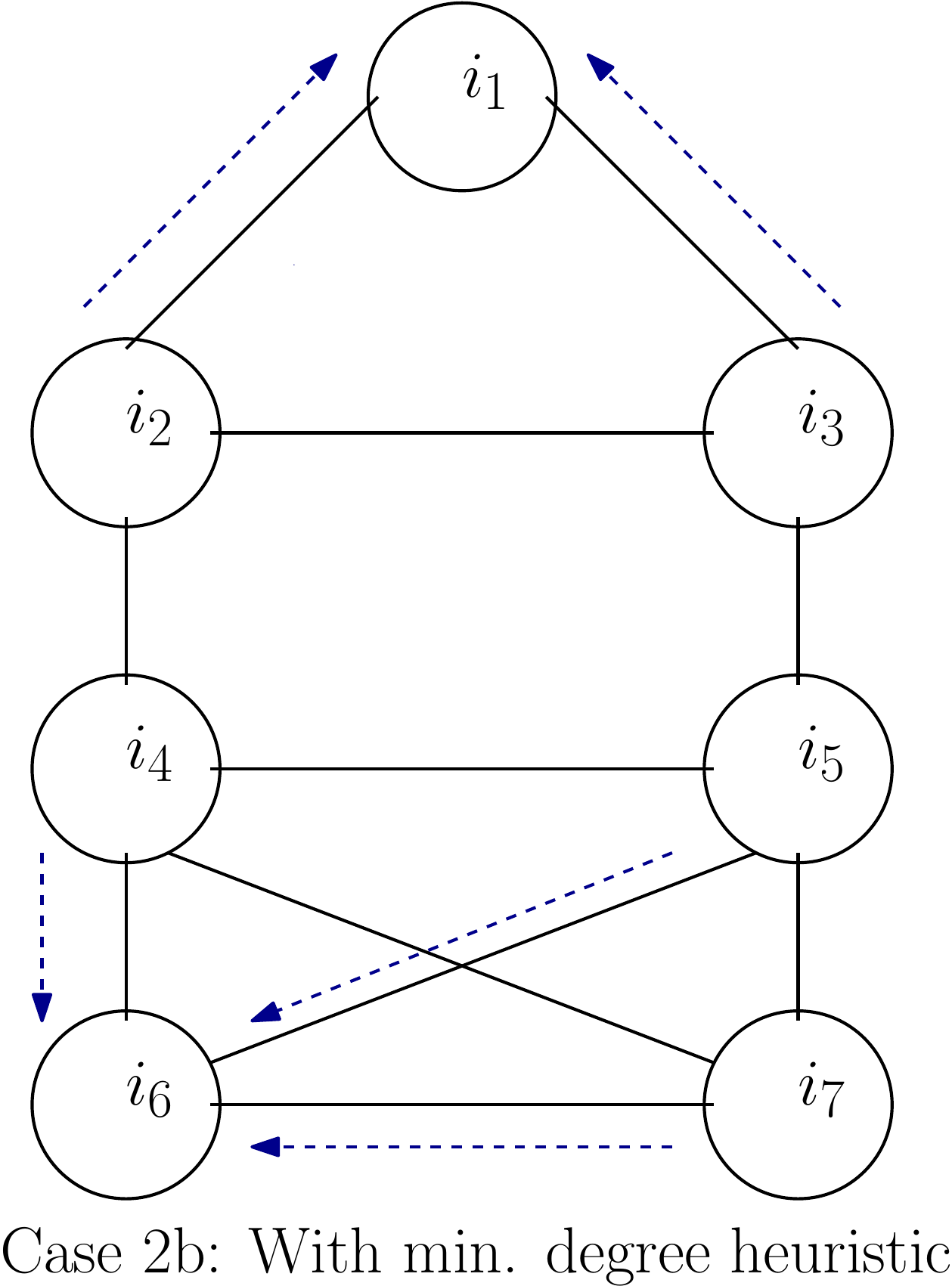} }
\caption{Examples of cases which can be handled by using the minimum labeling heuristic. The dotted arrows point to the direction of the vertex migration. 
Case 1 shows a scenario of vertex swap between two communities. 
%case where there are two vertices $i$ and $j$ such that $|C(i)|=|C(j)|=1$. Case 1a shows a scenario where they both swap their respective communities in parallel, as each of them select the other as the community target. If a minimum label heuristic is used, then only one of them will move to the other as shown in Case 1b.
Case 2) shows the evolution of two different communities $\{i_1,i_2,i_3\}$ and $\{i_4,i_5,i_6,i_7\}$. Without the application of any heuristic (Case 2b), the algorithm may either form partial communities (e.g., $\{i_1\},\{i_2,i_3\}$) or may settle on a local maxima (e.g., $\{i_4,i_6\},\{i_5,i_7\}$). Whereas the use of a minimum label heuristic could help the communities converge to the final solutions faster (as shown in Case 2b). 
}
\label{figMinLabel}
\end{figure*}

Section~\ref{secSwap} elaborated on the possibilities of swapping conditions that may delay the parallel algorithm's convergence to a solution. In this section we present a heuristic designed to address some of these cases.
Let us consider the simple case of two vertices $i$ and $j$ outlined in Section~\ref{secSwap}. Here both vertices are initially in communities of size one, and a decision in favor of merging at any given iteration will lead them to simply swap their respective communities without resulting in any net modularity gain. This is outlined in the Case 1a of Figure~\ref{figMinLabel}. Such a swap can be easily prevented by introducing a labeling scheme where it can be enforced that only one of them move to other's community. 
More specifically, let the communities at any given stage of the algorithm be labeled numerically (in an arbitrary order).  We will use the notation $\ell(C)$ to denote the label of a community $C$. Then the heuristic is as follows: 

{\bf The singlet minimum label heuristic: } In the parallel algorithm, at any given iteration, if a vertex $i$ which is in a community by itself (i.e., $C(i)=\{i\}$), decides (in the interest of modularity gain) to move to another community $C(j)$ which also contains only one vertex $j$, then that move will be performed \emph{only if} $\ell(C(j))<\ell(C(i))$. 

%This heuristic prevents single-member communities from swapping their vertices as that will \emph{not} result in a net modularity gain.

%Note that the above heuristic is directional, in that it does \emph{not} guarantee that $j$ will move to $i'$s community if $i$ does not move. In other words, $j$ will move to $C(i)$ only if the latter yields the maximum modularity gain from $j'$s perspective.  In either case though, swapping will be prevented with this heuristic. 

The above heuristic can be generalized to other cases of swapping or local maxima. For instance, let us consider the 4-clique of $\{i_4,i_5,i_6,i_7\}$ shown in Figure~\ref{figMinLabel}: case 2, assuming that each vertex is in its own individual community to start with. Here, in the absence of an appropriate heuristic there is a chance that the algorithm would settle on a local maxima. For instance, maximum modularity gains can be achieved at vertex $i_4$ by either moving to $C(i_6)$ or $C(i_7)$, and similarly for vertex $i_5$. However, if $i_4$ moves to $C(i_6)$ and $i_5$ to $C(i_7)$, then the resulting solution $\{i_4,i_6\},\{i_5,i_7\}$  (shown in case 2a of Figure~\ref{figMinLabel}) will represent a local maxima from which the algorithm may not proceed in the current phase. This is because, once these partial communities form, there is no incentive for $i_4$ or $i_6$ to individually move to the community containing  $\{i_5,i_7\}$, without each other's company. This is a limitation imposed by the Louvain heuristic, which makes decisions at the vertex level. However, if we label and treat the communities in a certain way then such local maxima situations can be avoided. 
%More specifically, let us again assume that the communities are labeled numerically in an arbitrary order.

{\bf The generalized minimum label heuristic: } In the parallel algorithm, at any given iteration, if a vertex $i$ has \emph{multiple} neighboring communities yielding the maximum modularity gain, then the community which has the minimum label among them will be selected as $i'$s destination community.

In the example for Figure~\ref{figMinLabel}:case 2, vertices $i_6$ and $i_7$ will both yield the maximum modularity gain for vertices $i_4$ and $i_5$. However, using the above minimum label heuristic, all three vertices $\{i_4,i_5,i_7\}$ will migrate to $C(i_6)$, while $i_6$ stays in $C(i_6)$ --- i.e., assuming $\ell(C(i_4))<\ell(C(i_5))<\ell(C(i_6))<\ell(C(i_7))$.

%In principle, any scheme to order the vertex labels will work. We used minimum degree labeling.

While swap situations may delay convergence, they can never lead to nontermination of the algorithm due to the use of a minimum required net modularity gain threshold to continue a phase.  
As for local maxima, a general proof that effects of elimination of local maxima cases progressively as the algorithm progresses is not possible due to the heuristic nature of algorithm. However, many situations, similar to those explained earlier in Section~\ref{secSwap}, typically get resolved in subsequent phases; this is because the representation of the individual sub-communities as meta-vertices is likely to lead them to merge with one another forming the containing communities eventually in the output.

\vspace*{-0.1in}
\subsection{Coloring}
\label{secColoring}

In this section, we explore the idea of graph coloring to address some of the parallelization challenges outlined in Section~\ref{secChallenges}. A distance-$k$ coloring of a graph is an assignment of colors to vertices such that no two vertices separated by a distance of at most $k$ are assigned the same color.
It should be easy to see that using distance-1 coloring to partition the vertices into color sets prior to the processing would prevent vertex-to-vertex swap scenarios. In this scheme, vertices of the same color are processed in parallel, and this is equivalent of guaranteeing that no two adjacent vertices will be processed concurrently. However, distance-1 coloring may not be adequate to address  other potential complications that may arise during parallelization (see Section~\ref{secNegativeGain}). 
\begin{corollary}
  \label{corollary2}
    Applying and processing the vertices in parallel by distance-1 coloring does \emph{not} necessarily preclude the possibility of negative modularity gains between iterations. 
\end{corollary}
\begin{proof}
  Follows directly from the three vertex example case presented for Lemma~\ref{lemma1}. 
\end{proof}
In fact the same result can be extended for application of a distance-k coloring scheme, where $k>1$, as was shown in \cite{lu_parallel_2014}.

Despite these lack of guarantees for a positive modularity gain between iterations, coloring still could be effective as a heuristic in practice, as we will demonstrate in Section~\ref{secExp}. 
%The cases outlined above in the proof for the Lemma~\ref{lemma3} are expected to be rare, given the number of conditions that need to be met. 
The performance trade-off presented by coloring is a potential reduction in the degree of parallelism versus faster convergence to higher modularity.
Coloring also presents an added advantage of being able to use higher modularity gain thresholds during the earlier phases of the algorithm, as will be explored in Section~\ref{secExp}.
%Within each iteration, the ideal parallelism is equal to the number of vertices. With coloring, this shrinks to the cardinality of smallest color, and the number of parallel steps increases to the number of colors. 
The run-time cost of coloring is expected to be dominated by the time spent within iterations; furthermore, for scalability in preprocessing, we use a parallel implementation to perform coloring \cite{catalyurek_graph_2012}.

%There exists an alternative to the vertex coloring approach --- one that uses community assignments to color vertices. More specifically, vertices from two neighboring communities are colored using different colors, so that they are processed at different parallel steps. This approach could prevent cases of vertex swaps happening between two communities (see Section~\ref{secSwap}) and may help in the overall convergence. Here again, however, the degree of parallelism will be reduced. In this paper, we implement only the first coloring heuristic presented above.

\vspace*{-0.1in}
\subsection{The vertex following heuristic}
\label{secMerging}

In this section, we will layout a particular property of the \emph{serial} Louvain algorithm in the way it treats vertices with single neighbors, and devise a heuristic around it. For the purpose of the lemma below, we will assume a version of Louvain algorithm which continues with iterations within a phase, until the communities stop changing.
We also distinguish between vertex $i$ being a \emph{single degree} vertex and a \emph{single neighbor} vertex --- the former is when the only edge incident on $i$ is $(i,j)$, whereas the latter is when $i$ could have up to two edges incident with $(i,j)$ being mandatory and $(i,i)$ being optional.

\begin{lemma}
\label{lemma3}
Given an input graph $G(V,E,\omega)$, let $i$ and $j$ be two different vertices such that $i$ is a single degree vertex with only one incident edge $(i,j)\in E$. Then, in the final solution $C(i)=C(j)$ --- i.e., $i$ should be part of the same community as $j$. 
\end{lemma}
\begin{proof}
%Given an input graph $G(V,E,\omega)$, let $i$ and $j$ be two vertices in $V$ such that $(i,j)\in E$ and $(i,k)\notin E$ for $k\neq j$,$k\neq i$. Note that this automatically implies the optional presence of $(i,i)\in E$. 
Consider any iteration $r$ in which vertices $i$ and $j$ are in two different communities --- i.e., $C(i)\neq C(j)$. 
During iteration $r$, the value of $\Delta Q_{i\rightarrow C(j)}$ will evaluate to the following:
\begin{eqnarray}
  \label{eqnClaimIItemp}
  \Delta Q_{i\rightarrow C(j)} &=& \frac{\omega(i,j)}{m} + \frac{ 2\cdot k_i\cdot a_{C(i)\setminus \{i\}}-2\cdot k_i\cdot a_{C(j)}  }{(2m)^2} \nonumber \\
   &\geq& \frac{\omega(i,j)}{m}  - \frac{2\cdot k_i\cdot a_{C(j)}}{(2m)^2}   \; (\because a_{C(i)\setminus \{i\}}\geq 0) \nonumber \\
     &=&  \frac{\omega(i,j)}{2m^2} \left( 2m  - \frac{k_i\cdot a_{C(j)}}{\omega(i,j)}  \right)  
     %\nonumber
\end{eqnarray}

Since vertex $i$ is a single degree vertex, $k_i=\omega(i,j)$. Therefore,
\begin{eqnarray}
  \label{eqnClaimII}
  \Delta Q_{i\rightarrow C(j)} &\geq& \frac{\omega(i,j)}{2m^2} \left( 2m  - a_{C(j)}  \right) 
\end{eqnarray}

Now, if $i$ were to decide \emph{against} moving to $C(j)$, $\Delta Q_{i\rightarrow C(j)}\leq 0$. 
Given that the above inequality (\ref{eqnClaimII}) is a lower bound for $\Delta Q_{i\rightarrow C(j)}$, and also that all edge weights are non-negative:
\begin{eqnarray} 
\label{eqnClaimIIb}
   \Rightarrow  2m  - a_{C(j)} &\leq  & 0 \nonumber\\
	 \Rightarrow 2m &\leq & a_{C(j)} 
\end{eqnarray}
But inequality~(\ref{eqnClaimIIb}) is \emph{not} possible because $a_{C(j)}\leq 2m$  for any community (by the definition in Eqn.\ref{eqnAC}) and in this case, since $i\notin C(j)$, $a_{C(j)}\leq (2m-\omega(i,j)) < 2m$. This implies that $i$ will have no choice but to move to $C(j)$ in iteration $r$. 
\end{proof}

We refer to the guarantee provided by the above lemma as the \emph{vertex following (VF) rule}. Note that it is guaranteed to hold only for single degree vertices in the input graph. 
The implication of this rule is that there is no need to explicitly make decisions on single degree vertices during the Louvain algorithm's iterations. Instead, we can preprocess the input such that all single degree vertices are merged {\it a priori} into their respective neighboring vertex.  More specifically, let $i$ be a single degree vertex with $j$ as its neighbor. Then, we remove vertex $i$ from the graph, and replace $j$  with a new vertex $j^\prime$, such that $\Gamma(j^\prime)=\{\Gamma(j)\setminus \{i\}\} \cup \{j^\prime\}$ and $\omega(j^\prime,j^\prime)=\omega(i,j)$ if $(j,j)\notin E$; and  $\omega(j^\prime,j^\prime)=\omega(j,j)+\omega(i,j)$ otherwise.  
%Note that if for some single neighbor vertex $i$, its neighbor $j$ is also a single neighbor, then the $i$ is merged with $j$ only if $i<j$ (as a convention).

This preprocessing not only could help reduce the number of vertices that need to be considered during each iteration, but it also allows the vertices that contain multiple neighbors (that tend to be the hubs in the networks) be the main drivers of community migration decisions. This is more important  under a parallel setting because if the single degree vertices were retained in the network the hub nodes could potentially gravitate temporarily toward one of their single degree mates, thereby delaying progression of solution or getting stuck in a local maxima. 

We could also extend the result of the Lemma~\ref{lemma3} to benefit cases where vertex $i$ is a single \emph{neighbor} vertex. The idea is similar to that of a k-core decomposition of the graph \cite{batagelj_generalized_2002}. Intuitively, during preprocessing, single neighbor vertices can be collapsed into their only neighboring vertex recursively until the negative component of the inequality~(\ref{eqnClaimIItemp}) starts to dominate its positive counterpart. Termination of this recursive merging can be implemented either by explicitly calculating both sides of the inequality~(\ref{eqnClaimIItemp}) or by estimating through other means via lower bounds or statistical thresholds. The idea is to lead to fast compression of chains within the input graph prior to application of the Louvain heuristic. We omit further details of this idea and for the purpose of this paper, we only consider the single degree version of the vertex following heuristic for implementation and experimental evaluation.

%%%%%%%%%%%%%%%%%%%%%%%%%%%%%%%%%%%%%%%%%%%%%%%%%%%%%%%%%%%%%%%%%%%%%%%%%%%%%%%%%%%%%
%\begin{figure}[!htb]
\begin{algorithm}[!htb]
\small
\caption{The parallel Louvain algorithm (a single phase).}
\label{algo.parLouvian}
\begin{algorithmic}[1]
\Procedure{Parallel Louvain}{$G(V,E,\omega), C$}

%\State {\tt Initialization}
\For {{\bf each} $i\in V$ in {\tt parallel}}
   \State $C(i) \gets \{i\}$; $\ell(C(i)) \gets i$
   \State $C_{int}^{i} \gets 0$ \Comment{{\footnotesize counter for the \#intra-community edges due to $i$}}
   \For{{\bf each} $j \in \Gamma(i)$} $C_{tot}^{i} \gets C_{tot}^{i} + \omega(i,j)$ %\Comment{{\footnotesize Adjacency list of $i$ is denoted by $\Gamma(v)$}}
       %\State $C_{tot}^{i} \gets C_{tot}^{i} + \omega(i,j)$%\Comment{{\footnotesize counter for the \#inter-community edges due to $i$ }}
   \EndFor
\EndFor
\State $Q_C \gets 0$; $Q_P \gets -\infty$ \Comment{{\footnotesize Current \& previous modularity}}
%\State $Q_P \gets 0$ \Comment{{\footnotesize Previous modularity}}
%\State $Q_P \gets -\infty$ %\Comment{{\footnotesize Previous modularity}}

\While{{\tt true}} \Comment{{\footnotesize Iterate until modularity gain becomes negligible.}}
   %\State \Comment{{\footnotesize  {\tt Stage-1:} For each vertex, compute the modularity gain from moving to a neighboring cluster.}}
   \For {{\bf each} $i\in V$ in {\tt parallel}}
      \State $C_{old} \gets C(i)$; $N_i\gets C(i)$%\Comment{{\footnotesize Neighboring communities of $i$}}
      \For{{\bf each} $j \in \Gamma(i)$} $N_i \gets N_i \cup C_j$
            %\Comment{{\footnotesize Find unique cluster ids and their edges}}
      \EndFor

      \State $maxGain \gets 0$; $C_{new} \gets C_{old}$ %\Comment{{\footnotesize Compute modularity gain from reclustering $v$.}}
      %\State $C_{new} \gets C_{old}$
      \For {{\bf each} $c\in N_i$ in {\tt parallel}}
        \State $curGain \gets$ Calculate $\Delta Q_{i\rightarrow c}$ 
           \If{$((curGain > maxGain)$ or $ (curGain=maxGain$ and $\ell(c)<\ell(C_{new}) )$}\Comment{{\footnotesize Minimum label heuristic}}
              \State $maxGain \gets curGain$; $C_{new} \gets c$
           \EndIf
      \EndFor

      \If{$maxGain > 0$}
         \State $C_{old} \gets C_{old} \setminus \{i\}$; $C_{new} \gets C_{new} \cup \{i\}$ 
      \EndIf

   \EndFor

 %\State \Comment{{\footnotesize {\tt Stage-2:} Compute the net modularity for this iteration.}}
 \For {{\bf each} $c \in C$ {\tt AND} $c \neq \emptyset$ in {\tt parallel}}
    \State $C_{int}^{c} \gets 0$; $C_{tot}^{c} \gets 0$
 \EndFor

 \For {{\bf each} $(i,j) \in E$ in {\tt parallel}}
     \If{$C(i) = C(j)$}
         \State $C_{int}^{i} = C_{int}^{i} + \omega(i,j)$
         \State $C_{tot}^{i} = C_{tot}^{i} + \omega(i,j)$
     \Else
         \State $C_{tot}^{i} = C_{tot}^{i} + \omega(i,j)$
         \State $C_{tot}^{j} = C_{tot}^{j} + \omega(i,j)$
    \EndIf
 \EndFor

 \State $e_{xx} \gets 0$
 \State $a^2_x \gets 0$
 \For {{\bf each} $c \in C$ {\tt AND} $c \neq \emptyset$ in {\tt parallel}}
      \State $e_{xx} += C_{int}^{c}$; $a^2_x  += (C_{tot}^{c})^{2}$
 \EndFor

 \State $Q_C = \frac{e_{xx}}{m} -  \frac{a^2_x}{(2m)^2}$

 \If{$\vert \frac{Q_C - Q_P}{Q_P} \vert < \theta$} \Comment{{\footnotesize $\theta$ is a user specified threshold.}}
      \State {\tt break}  \Comment{{\footnotesize Phase termination}}
 \Else
      \State $Q_P \gets Q_C$ 
 \EndIf

\EndWhile
\EndProcedure
\end{algorithmic}
\end{algorithm}
%\caption{
%Pseudocode for each phase of our parallel Louvain method.
%}
%\label{figParAlgo}
%\end{figure}
%%%%%%%%%%%%%%%%%%%%%%%%%%%%%%%%%%%%%%%%%%%%%%%%%%%%%%%%%%%%%%%%%%%%%%%%%%%%%%%%%%%%%

\subsection{Parallel algorithm }
\label{secParallelAlgo}

Our parallel algorithm has the following major steps:

\begin{compactenum}[1)]
  \item{\underline{VF preprocessing (Optional):}} Apply the vertex following heuristic by merging all single degree vertices into their respective neighboring vertices  (as explained in Section~\ref{secMerging}). This step is performed in parallel. Label the resulting vertices from $1\ldots n$ using an arbitrary ordering. 
    \item{\underline{Coloring preprocessing (Optional):}} Color the input vertices using distance-k coloring. For this paper, we only explore distance-1 coloring. For coloring, we used the parallel implementation from \cite{catalyurek_graph_2012}. 
  \item{\underline{Phases:} } Execute phases one at a time as per Algorithm~\ref{algo.parLouvian}. Within each phase, the algorithm runs multiple iterations, with each iteration performing a parallel sweep of vertices without locks and using the community information available from the previous iteration. If coloring was applied, then the processing of each color set is parallelized internally and the community information from the previous coloring stages is available to make migration decisions in subsequent coloring stages. This is carried on until the modularity gain between successive iterations becomes negligible. 
  %The iterations are the main parallel steps, and the  algorithm for each iteration is also shown within Algorithm~\ref{algo.parLouvian}.  
  \item{\underline{Graph rebuilding:} } Between two successive phases, the community assignment output of the completed phase is used to construct the input graph for the next phase. This is done by representing all communities of the completed phase as ``vertices'' and accordingly introducing edges, identical to the manner in which it is done in the serial algorithm.    This step is also implemented in parallel as described in Section~\ref{secImplementation}.
\end{compactenum}
 
%While Algorithm~\ref{algo.parLouvian} is for shared memory multicore architectures, the same set of heuristics and ideas carry over to distributed memory parallelism except for the implementation. 

%In fact, we have two implementations based on this algorithm --- one that uses OpenMP for multicore platforms and another a MapReduce implementation for distributed memory clusters using the MapReduce-MPI library \cite{plimpton_mapreduce_2011}.

We note here that the above parallel algorithm, with the exception of coloring heuristic, is stable in that it always produces the same output regardless of the number of cores used.
%This is owing to the fact that at each iteration the decisions made by every vertex are based on the state of communities from the previous iteration (proof trivial).
When coloring is applied, the use of multiple threads within a given iteration could potentially vary the order in which decisions are made, thereby leading to potential variations in the output. In our experiments, we found the magnitudes of such variations to be negligible.

\subsection{Implementation}
\label{secImplementation}

We implemented our parallel heuristics in C++ and OpenMP. 
It is to be noted that the heuristics themselves are agnostic to the underlying  parallel architecture. 
There are a few implementation level variations to Algorithm~\ref{algo.parLouvian}. In Algorithm~\ref{algo.parLouvian} the modularity calculation happens in lines $27-31$. The steps from $18-26$ calculates the intra- and inter-community edge counts required for the modularity calculation. In our actual implementation we do not explicitly execute these steps. Instead we pre-aggregate these values in steps $9-17$ as the net modularity gains are being calculated for each vertex. 
This saves significant rework. 
Secondly, lines $16-17$ show the update for the source and target communities for each vertex $i$. 
We implemented these updates using intrinsic atomic operations {\tt \_\_sync\_fetch\_and\_add()} and 
{\tt \_\_sync\_fetch\_and\_sub()}. 

We use a compressed storage format for graph data structures that store the adjacency lists for all the vertices in a contiguous memory location.
Specific memory pointers for each vertex is maintained in a separate list. 
This format enables efficient access to neighborhood information for each vertex.
We use the C++ STL \texttt{map} data structure to store the set of unique clusters that a vertex is connected to (i.e., neighboring communities). The number of possible choices is upperbounded by the degree of a vertex initially and depending on how fast the algorithm converges from iteration to iteration, the number of choices decreases. Since this step appears in the computation for every vertex, we also experimented with several alternatives including the use of C++ STL \texttt{unordered\_map} data structure, but did not find any significant improvements in performance.
%Improving the efficiency of this operation is a key part of our future work.

The step to rebuild the graph between consecutive phases is implemented in parallel and serial in parts.
This is achieved in a sequence of steps.  Assume that the phase transition is between phase $i-$ to $i$. We use $G_{i-1}$ and $G_{i}$ to refer to the graphs input to phases $i-1$ and $i$ respectively.
i) First, the set of vertices in $G_{i}$ is constructed from the communities output from phase $i-1$. Since many communities which existed at the start of phase $i-1$ could have become empty by the end of that phase, we first renumber of communities numerically, using only non-empty communities. This step is currently implemented in serial, although our future plan is to explore a parallelization using prefix computation-based approach. 
ii) In the next step, a STL \texttt{map} structure is allocated for every new vertex in $G_{i}$ to concisely store the set of neighboring communities attached to it. This step is parallel.
iii) In the following step, all edges in $G_{i-1}$ are traversed in parallel. If an edge is an intra-community edge, then the weight for the corresponding edge (connecting the community vertex to itself) in $G_{i}$ is updated. Alternatively, an inter-community edge leads to an update to each of the two corresponding community vertices in $G_{i}$. The former requires one lock and the latter requires two.

Our implementation is named {\em Grappolo}\footnote{Italian word meaning a cluster (of grapes)}. The software is available for download under the BSD 3-Clause license from here: 
\url{http://hpc.pnl.gov/people/hala/grappolo.html}. 

\subsection{Analysis}
\label{secAnalysis}

Within each  iteration (refer to Algorithm~\ref{algo.parLouvian}), the vertices are scanned in parallel, and for every vertex their vertex neighborhood is scanned first to curate the set of distinct neighboring communities (steps $9-10$). Subsequently, the main step of modularity gain calculation is performed only for each distinct neighboring community (steps $12-17$), which is equal to vertex degree initially but is expected to rapidly reduce as the iterations progress. Consequently, the worst-case runtime complexity \emph{per} iteration is $O(\max\{{\frac{M+n\cdot \lambda}{p},\lambda_{max}}\})$, where $p$ denote the number of processing cores, $\lambda$ is the average (unweighted) degree of a vertex and $\lambda_{max}$ is the maximum (unweighted) degree of a vertex. The space complexity is linear in the input for shared memory implementation (i.e., $O(m+n)$).
The above analysis assumes that the entire collectin of vertices is processed in one parallel step within each iteration. 
With the application of coloring, parallelism is limited to each color set, implying the number of color sets to correspond to the number of parallel steps within each iteration.
%, whereas $O(\frac{m+n}{p})$ for the distributed memory implementation.

\section{Experimental evaluation}
\label{secExp}

\subsection{Experimental setup}
\label{secExpSetup}

%{\it MAHANTESH TO VERIFY}

The test platform for our experiments is an Intel Xeon X7560 server with four sockets and $256$ GB of memory. Each socket is equipped with eight cores running at $2.266$ GHz, leading to a total of $32$ cores. The system is equipped with $32$ KB of L1 and $256$ KB L2 caches per core, and $24$ MB of cache per socket. Each socket has $64$ GB of DDR3 memory with a peak bandwidth of $34.1$ GB per second. 
The software was compiled with GCC version 4.8.2 using \texttt{-Ofast} option. We also enabled non-uniform memory distribution using \texttt{numactl} command and enabled thread binding by using \texttt{GOMP\_CPU\_AFFINITY} environment variable. The thread binding variable was configured to place the threads across the system as evenly as possible with the goal of maximizing the memory bandwidth. All experiments were run using one thread per core.

We tested our heuristics on $11$ different real world input graphs, which are summarized in Table~\ref{tabInput}. With the exception of inputs labeled ``MG1'' and ``MG2'', all other inputs were downloaded from the DIMACS10 challenge website \cite{dimacs10_10th_????,bader_graph_2012}, and the University of Florida sparse matrix collection~\cite{davis_university_2011}. ``MG1'' and ``MG2'' are graphs constructed for two different ocean metagenomics data, using the construction procedure described in \cite{wu_pgraph:_2012}.

\begin{table}[tbh]
\caption{\label{tabInput}
Input statistics for the real world networks used in our experimental study.
``RSD'' represents the relative standard deviation of vertex degrees for each graph. 
It is given by the ratio between the standard deviation of the degree and its  mean.
}
\centering
\begin{tabular}{|l|r|r|r|r|r|}
\hline
Input &Num. &Num. & \multicolumn{3}{c|}{Degree statistics ($\lambda$)} \\
graph &vertices (n)&edges (M)& max. & avg. & RSD \\
\hline
\hline
CNR & 325,557 & 2,738,970 & 18,236 & 16.826 & 13.024 \\
\hline
coPapersDBLP & 540,486 & 15,245,729 & 3,299 & 56.414 & 1.174 \\
\hline
Channel & 4,802,000 & 42,681,372 & 18 & 17.776 & 0.061  \\
\hline
Europe-osm & 50,912,018 & 54,054,660 & 13 & 2.123 & 0.225 \\
\hline
Soc-LiveJournal1 & 4,847,571 & 68,475,391 & 22,887 & 28.251 & 2.553\\
\hline
MG1& 1,280,000 & 102,268,735 & 148,155 & 159.794 & 2.311\\
\hline
Rgg\_n\_2\_24\_s0 & 16,777,216 & 132,557,200 & 40 & 15.802 & 0.251\\
\hline
uk-2002 & 18,520,486 & 261,787,258 & 194,955 & 28.270 & 5.124 \\
\hline
NLPKKT240 & 27,993,600 & 373,239,376 & 27 & 26.666 & 0.083\\
\hline
MG2 & 11,005,829 & 674,142,381 & 5,466 & 122.506 & 2.370 \\
\hline
friendster & 51,952,104 & 1,801,014,245 & 8,603,554 & 69.333 & 17.354\\
\hline
\hline
\end{tabular}
\end{table}

The input graphs were tested using multiple variants of our implementation that use different combination of the proposed heuristics. These variants are as follows:
\begin{itemize}
\item{\bf baseline:} represents our parallel implementation with only the Minimum Labeling (ML) heuristic;
\item{\bf baseline+VF:} represents the baseline implementation with the application of the Vertex Following (VF) heuristic  in a preprocessing step.
There were a few inputs (viz., Channel, MG1, MG2) for which their single degree vertices had already been pruned off when their respective graphs were generated, and consequently their baseline runs are equivalent to their  baseline+VF runs\footnote{For this reason, we show only their baseline+VF runs in their respective charts.}. 
For the remaining inputs, VF preprocessing was run only once, prior to the start of the first phase;
%\item{\bf baseline+Color:} represents the baseline implementation with the application of the coloring heuristic applied in a preprocessing step;
%\footnote{In all our experiments, we applied coloring only for preprocessing the input for the first phase.} 
\item{\bf baseline+VF+Color:} represents the baseline implementation with the application of both the VF and coloring heuristics (in that order).
Coloring was used as a preprocessing step for multiple phases until either the number of input vertices reduced below a preset cutoff (100K used for this paper) or the net modularity gain between phases is less than the user-defined threshold ($10^{-2}$).
Once either of these conditions is met, the implementation does not perform coloring anymore and the remaining phases are executed using a default net modularity gain threshold of $10^{-6}$ for termination.
%We experimented with $10^{-2}$ and $10^{-4}$ for this paper.
\end{itemize}

\subsection{Performance evaluation}
\label{secPerf}

To assess the effectiveness of our parallel heuristics, we studied how quickly a given algorithm converges to its final modularity (as a function of the number of iterations) and compared it against the convergence rate of the corresponding serial Louvain\footnote{All references to the ``serial'' implementation in the experimental results section corresponds to the original Louvain implementation available from \cite{louvain_findcommunities_????}.} execution. 
We also compared the difference in runtimes and final modularities output by the individual approaches. 
Figures~\ref{figModCurves}-\ref{figModCurves4} show the evolution of modularity from the first iteration of the first phase to the last iteration of the last phase for all the 11 test inputs, and the parallel runtimes as a function of the number of cores.  

%It is to be noted that this study was conducted for all 11 test inputs. 
%However, in the interest of space, we show results for only a selected subset of these test inputs, leaving out those inputs for which the general trends were similar to at least one of the other inputs shown. We have also tried to highlight any observed deviations in the text.

%%%%% a-c

\begin{figure}[!htb]
\centering
\includegraphics[width=2.35in]{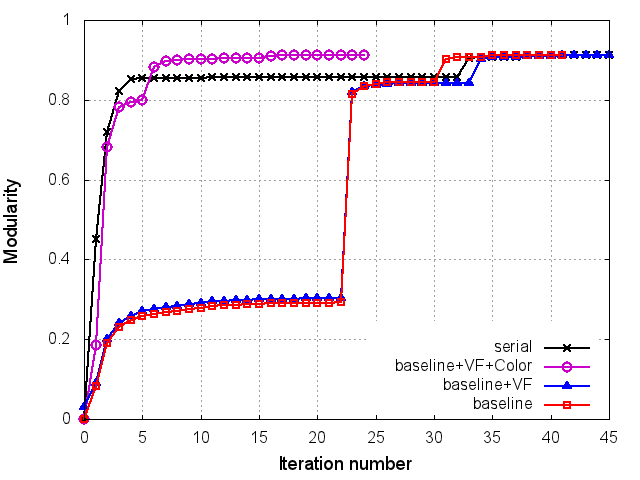}
%\hspace*{0.5in}
\includegraphics[width=2.35in]{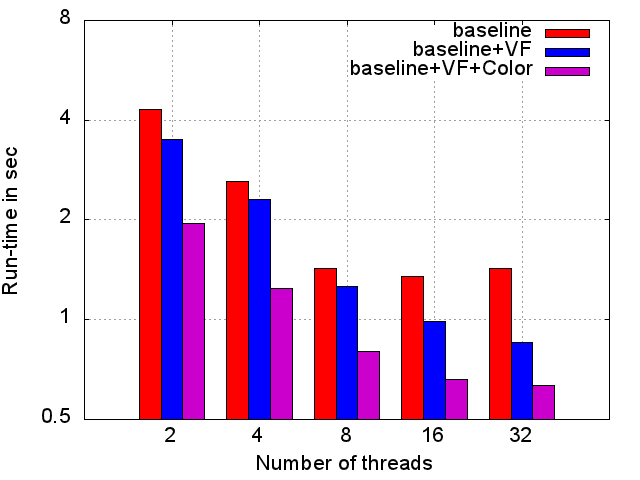}
%\hspace*{0.5in}
\\
{\footnotesize
(a) Input: CNR 
}
\\
\includegraphics[width=2.35in]{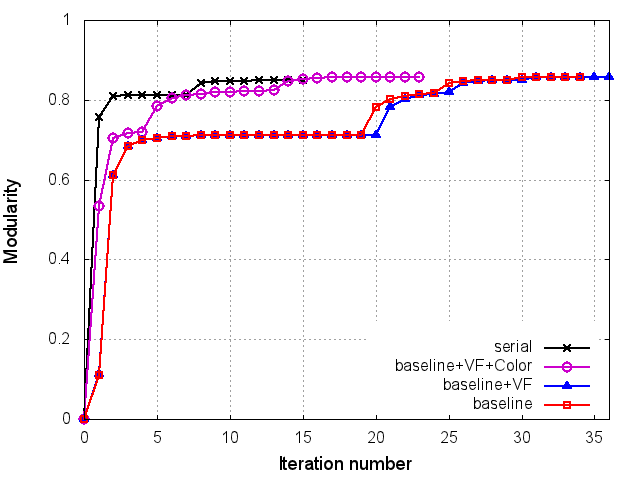}
\includegraphics[width=2.35in]{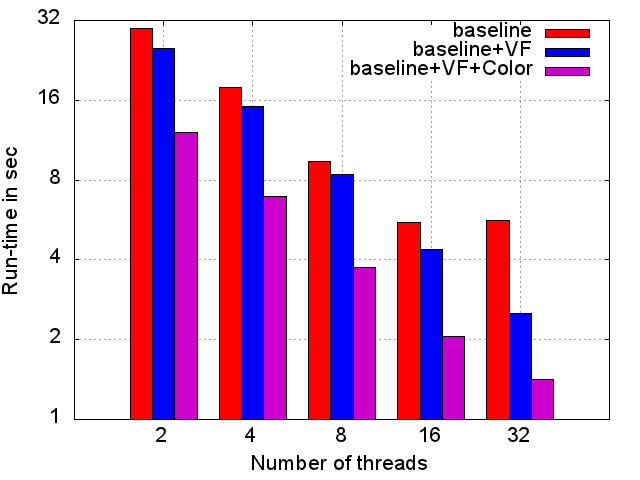}
%\hspace*{0.5in}
\\
{\footnotesize
(b) Input: coPapersDBLP
}
\\
\includegraphics[width=2.35in]{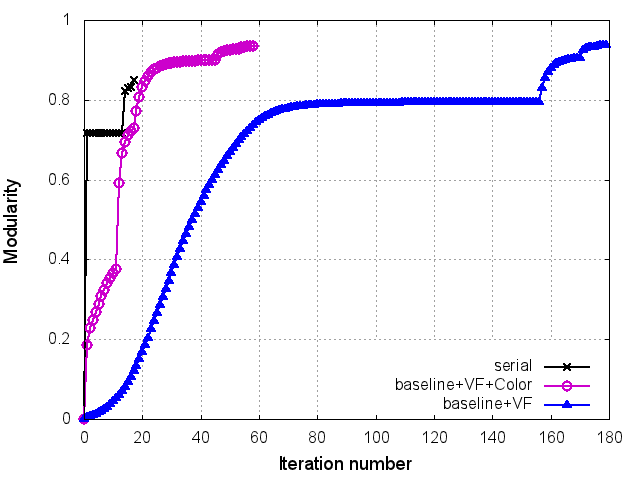}
\includegraphics[width=2.35in]{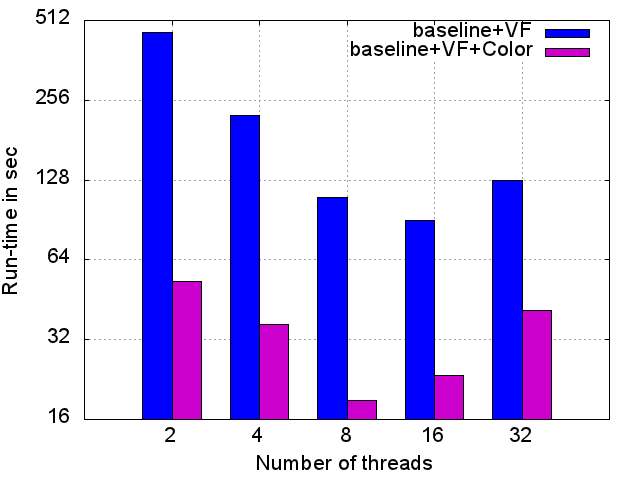}
%\hspace*{0.5in}
\\
{\footnotesize
(c) Input: Channel
}
\caption{ 
Charts showing the evolution of modularity (left column) and the parallel runtime performance (right column) for each test input.  
The steep climbs in modularity visible in the modularity curves correspond to phase transitions. 
Also shown for comparison are the corresponding performance of the serial algorithm.
}\label{figModCurves}
\end{figure}

%%%%% d-f

\begin{figure}[!htb]
\centering
\includegraphics[width=2.35in]{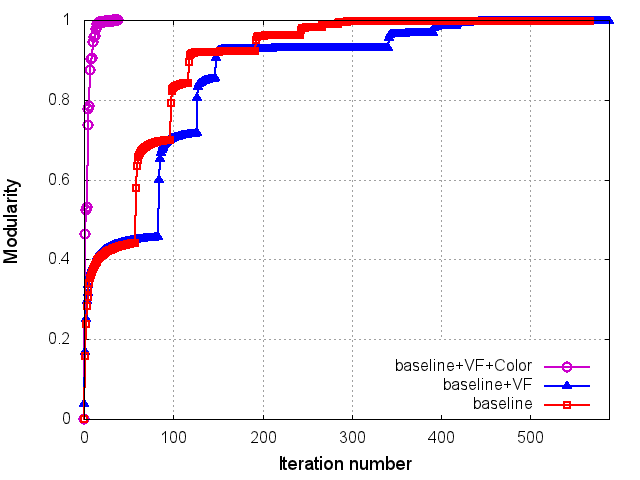}
\includegraphics[width=2.35in]{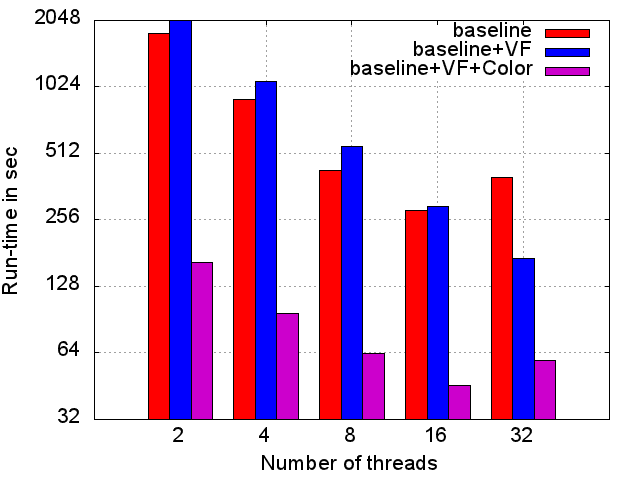}
%\hspace*{0.5in}
\\
{\footnotesize
(d) Input: Europe-osm
}
\\
\includegraphics[width=2.35in]{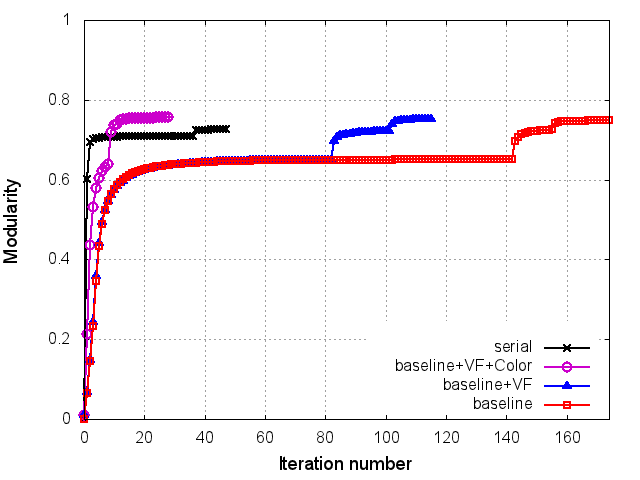}
\includegraphics[width=2.35in]{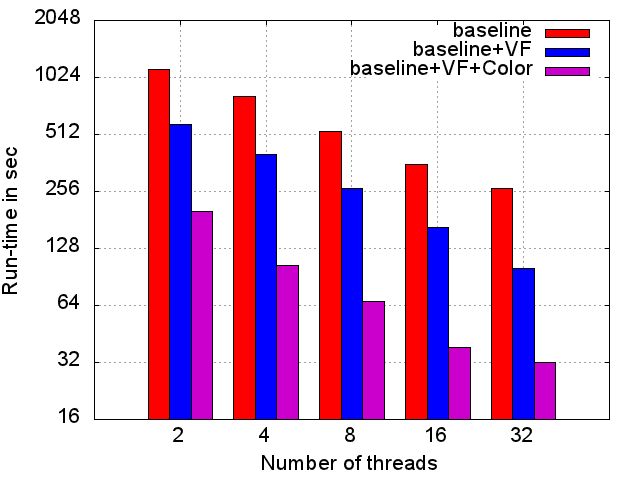}
%\hspace*{0.5in}
\\
{\footnotesize
(e) Input: Soc-LiveJournal1
}
\\
\includegraphics[width=2.35in]{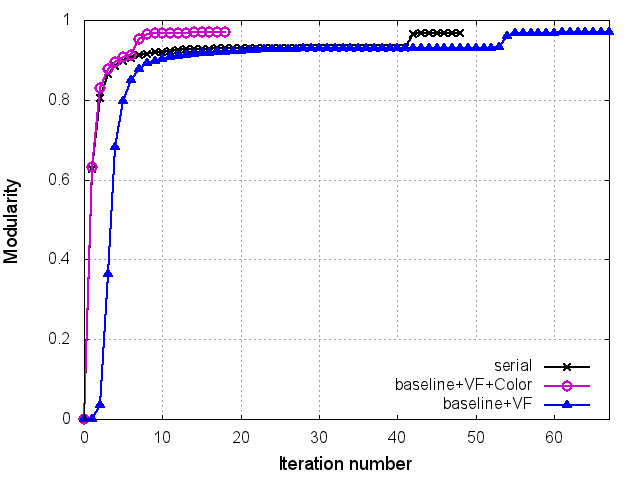}
\includegraphics[width=2.35in]{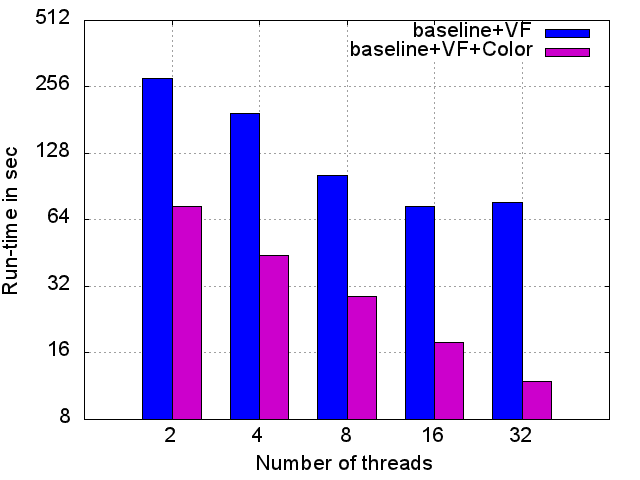}
%\hspace*{0.5in}
\\
{\footnotesize
(f) Input: MG1
}
\caption{ 
Charts showing the evolution of modularity (left column) and the parallel runtime performance (right column) for each test input.  
}\label{figModCurves2}
\end{figure}

%%%%% g-i

\begin{figure}[!htb]
\centering
\includegraphics[width=2.35in]{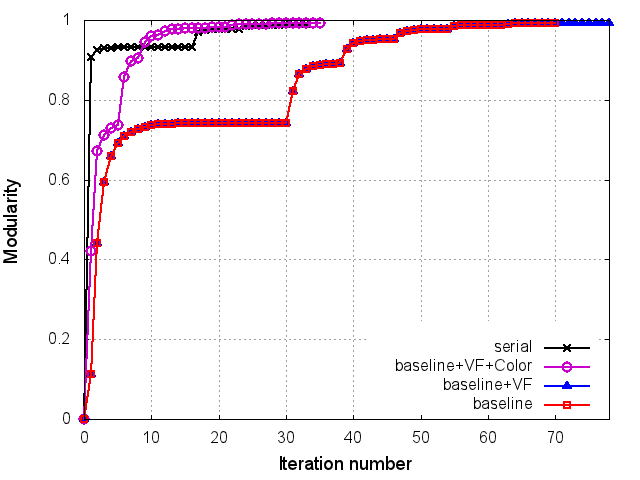}
\includegraphics[width=2.35in]{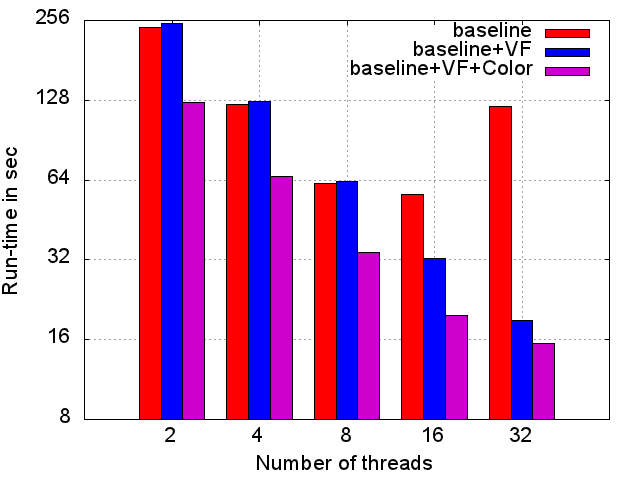}
%\hspace*{0.5in}
\\
{\footnotesize
(g) Input: Rgg\_n\_2\_24\_s0 
}
\\
\includegraphics[width=2.35in]{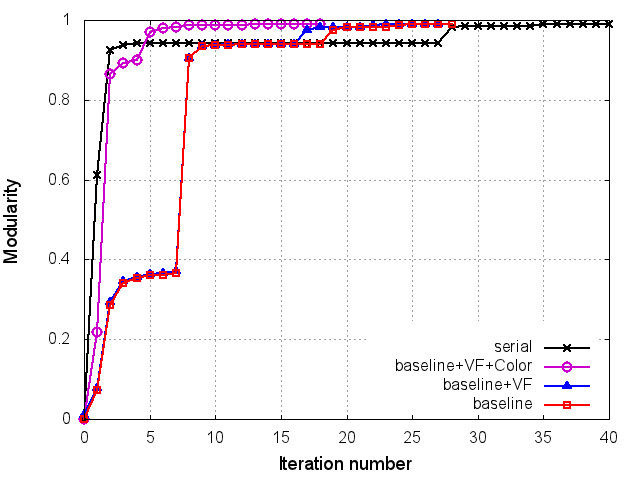}
\includegraphics[width=2.35in]{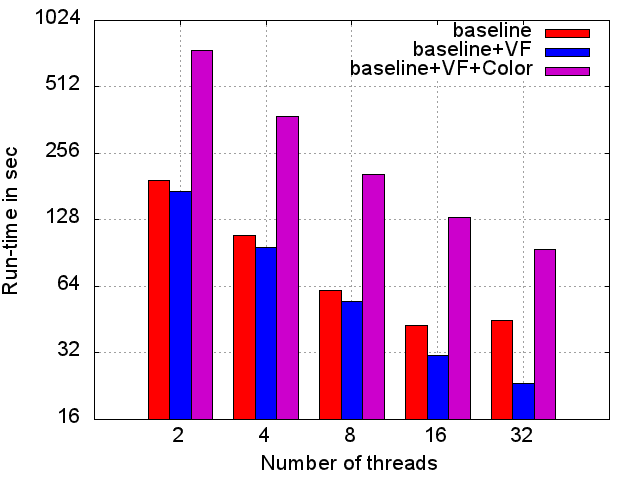}
%\hspace*{0.5in}
\\
{\footnotesize
(h) Input: uk-2002
}
\\
\includegraphics[width=2.35in]{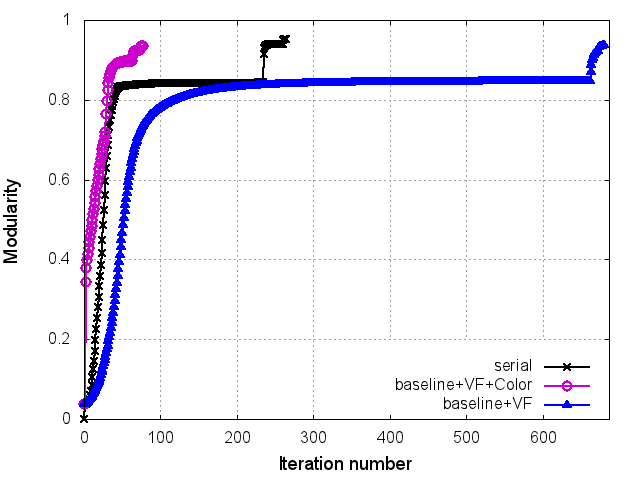}
\includegraphics[width=2.35in]{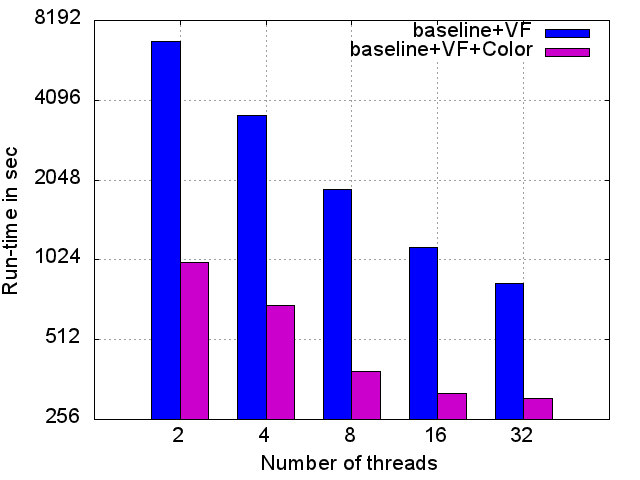}
%\hspace*{0.5in}
\\
{\footnotesize
(i) Input: NLPKKT240
}
\caption{ 
Charts showing the evolution of modularity (left column) and the parallel runtime performance (right column) for each test input.  
}\label{figModCurves3}
\end{figure}

%%%%% j-k

\begin{figure}[!htb]
\centering
\includegraphics[width=2.35in]{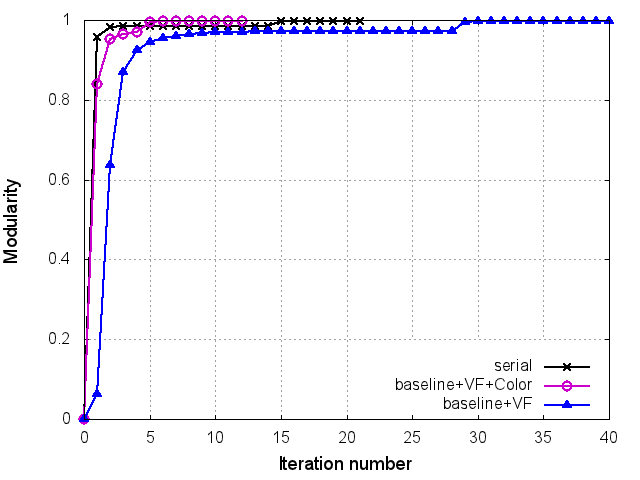}
\includegraphics[width=2.35in]{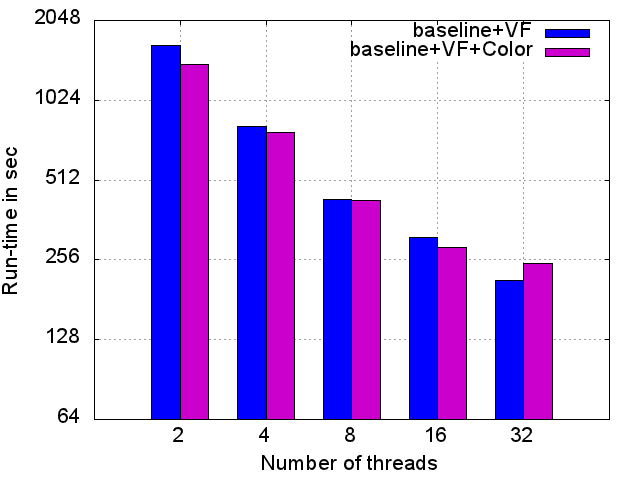}
%\hspace*{0.5in}
\\
{\footnotesize
(j) Input: MG2
}
\\
\includegraphics[width=2.35in]{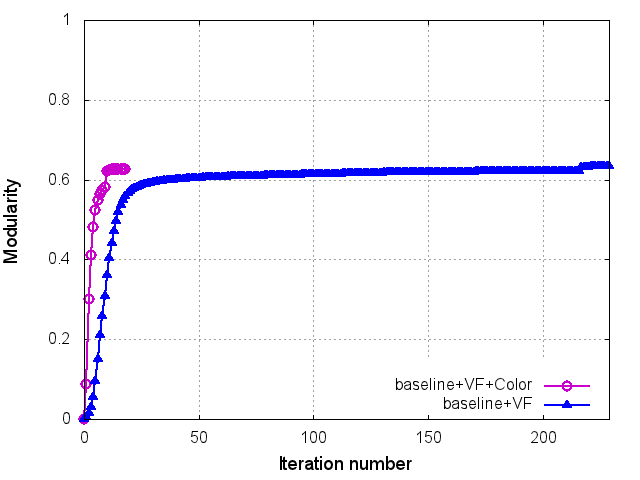}
\includegraphics[width=2.35in]{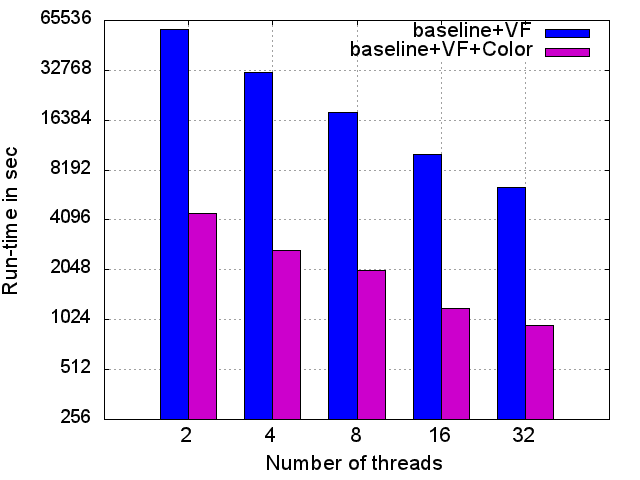}
%\hspace*{0.5in}
\\
{\footnotesize
(k) Input: friendster
}
\caption{ 
Charts showing the evolution of modularity (left column) and the parallel runtime performance (right column) for each test input.  
}\label{figModCurves4}
\end{figure}

{\bf Effectiveness of the VF heuristic:} 
The run-time charts in Figures~\ref{figModCurves}-\ref{figModCurves4} show the effectiveness of the VF heuristic in reducing run-time relative to our baseline implementation.
The reduction in run-time can be attributed to the reduction in the number of vertices to be processed within each iteration. However, the effectiveness of the VF heuristic is also tied to the number of single degree vertices in the original input graph.
While our results show that VF is able to produce run-time savings in most input cases, there were two exceptions: Europe-osm (Figure~\ref{figModCurves2}d) and Rgg\_n\_2\_24\_s0 (Figure~\ref{figModCurves3}g), for which the run-time was observed to increase. Upon further investigation, we found that the application of VF for these two inputs indeed caused a reduction in the time spent per iteration as expected; however, it also led to prolonging the convergence of the algorithm within the intial phases --- i.e., it led to an increase in the number of iterations within a phase. 

This delay in convergence within a phase shows a potential drawback of the VF heuristic on some input cases that can be intuitively explained as follows: consider a chain of ``hub'' nodes where the hubs are individually connected to a number of single degree vertices (``spokes''). In such cases, the compacted representation that results from the application of VF would have more incentive to continue in the current phase by gradually collapsing the chain into larger communities and achieving smaller gains in modularity that still surpass the minimum required cutoff. This results in prolonging the termination of the current phase. In contrast, if we were to omit applying the VF heuristic on the input graph, then a hub node could potentially migrate into one of its spokes' communities and when that happens, there is an increased probability that the algorithm terminates the current phase sooner due to negligible modularity gain. While the resulting final modularity figures could be slightly lower than obtained with the application of VF, the gains in runtime may be more pronounced, which is what we observed for the two inputs Europe-osm and Rgg\_n\_2\_24\_s0. It is to this end, that the proposed extension of the VF heuristic that also seeks to compress paths (see discussion at the end of Section~\ref{secMerging}) could aid in obtaining a better balance between run-time benefit and modularity gain.

%Coloring
{\bf Effectiveness of coloring:} 
The design intent of coloring is to reduce the number of iterations required to converge on a solution, and in the process reduce the time to solution. However, a potential drawback of coloring is reduced parallelism within each iteration; more specifically, the presence of numerous small color sets could result in an under-utilization of threads. In our experimental results, we found coloring to be highly effective in reducing both the number of iterations \emph{and} the overall time to solution. The run-time improvements of \emph{baseline+VF+coloring} over \emph{baseline+VF} were anywhere from $\sim 3.48\times$ to $16.52\times$.  
However, the run-time improvements were either negligible in the case of MG2 (Figure~\ref{figModCurves4}j) or negative in the case of uk-2002 (Figure~\ref{figModCurves3}h). 
These observations correlate with the highly skewed color size distributions for these two graphs. For instance, 943 colors were used for uk-2002 in the first phase and the color sets had a high Relative Standard Deviation (RSD) of 18.876 in their sizes. We are exploring an alternative approaches to create balanced coloring sets that are targeted at addressing this performance issue. For all other inputs, however, the benefit of coloring is evident in the drastically reduced number of iterations for convergence and subsequent savings in the time to solution. 
These results also show the combined effect of applying both VF and coloring heuristics, as they yield an additive net gain in performance.

%{\it Dr. Ananth to modify} In Figures~\ref{figColorInIteration}, it show that within the first phase, before phases transition, coloring move the modularity curve closer to the serial line. This allowing the early phase transition without bad clusters.

\subsubsection{Scaling and run-time results}
\label{secScaling}
Figure~\ref{figSpeedup} shows the speedup curves for our parallel implementation ({\it baseline+VF+Color}). Two speedup curves are shown: a) \emph{relative speedup}, which calculates the speedup of the parallel execution over the corresponding 2-thread run (discussed in this section); and b) \emph{absolute speedup}, which is the speedup calculated over the corresponding serial Louvain implementation's execution \cite{louvain_findcommunities_????} (to be discussed in Section~\ref{secSerialCompare}).

The relative speedup curves show that on most inputs, the parallel implementation continues to deliver increasing speedups up to 32 threads, although the speedups become sub-linear beyond 8 threads.
While the input sizes play a role, it can be observed from the results that the size alone is not the sole determinant of performance.
For instance, the implementation achieves higher peak relative speedups ($\sim$$8\times$) on some of the smaller inputs such as coPapersDBLP (540K vertices, 15M edges) and Rgg\_n\_2\_24\_s0 (16M vertices, 132M edges) than on a larger input such as NLPKKT240 (51M vertices, 1.8B edges). 
Parallel performance is affected by a combination of input characteristics and the serial bottlenecks within the parallel implementation.

Inputs Channel and NLPKKT240 have a low RSD in vertex degree distribution (Table~\ref{tabInput}) and also have a poor community structure (reflected in their low modularity scores). 
This combination leads to an increased number of iterations in the initial phases, as the algorithm continues within a phase albeit incremental modularity gains.
The increased number of iterations in the first phase in particular (where the graph size is the largest) adversely affects on performance. 
This is because within each iteration the step to recalculate the new modularity score involves updating community structures (internal edge and incident edge counts; lines $20-26$ in Algorithm~\ref{algo.parLouvian}); and as the number of communities begins to reduce in the later iterations of a phase, more parallel overhead due to locking is incurred. 

In contrast, consider the input $Rgg\_n\_2\_24\_s0$ which also has a low RSD in its vertex degree distribution but for which a superior parallel performance is observed. This input is a random geometric graph, which despite its uniform degree distribution, is also known to have a high community structure (reflected by its high modularity score). This attribute allows the algorithm to rapidly converge within the first phase, thereby aiding better overall parallel performance.

%Another input attribute that affects parallel performance is the number of color sets in the input graph. Note that the parallel algorithm processes vertices of one color set at a time. Therefore, a large number of color sets could negatively impact performance.

\begin{figure}[!thb]
\centering
\includegraphics[width=2.6in]{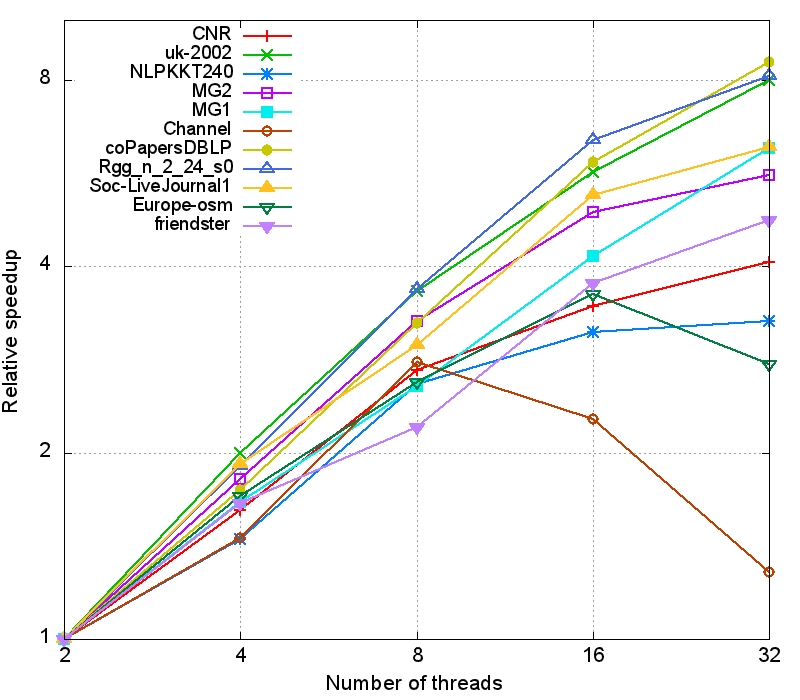}
\includegraphics[width=2.6in]{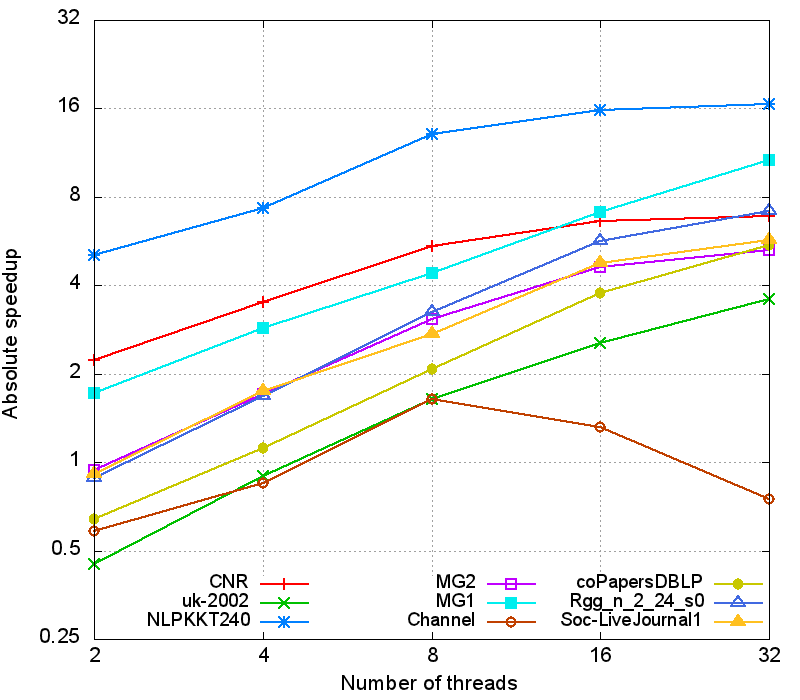}
\caption{ 
Speedup charts for our parallel implementation, {\it Grappolo}. The chart on left shows the relative speedup of the parallel implementation using the 2-thread run as the reference. The chart on the right shows the absolute speedup --- i.e., relative to the serial Louvain implementation \cite{louvain_findcommunities_????}.
All speedups are calculated using the {\it baseline+VF+Color} implementation of {\it Grappolo}.
Note that in the absolute speedup chart, curves for Europe-osm and friendster are not shown because the serial Louvain implementation failed to complete on these two inputs.
}\label{figSpeedup}
\end{figure}

%%%%%%%%%%%%%%OLD TEXT BEFORE REVISION
%Figures~\ref{figModCurves}-\ref{figModCurves4}  also show the run-time scaling of our parallel implementations as a function of the number of cores. As can be expected, for smaller inputs such as CNR linear scaling is observed for fewer (8) threads, and for the larger inputs (e.g., Rgg\_n\_2\_24\_s0, friendster) linear scaling extends to more (16 or more) threads - e.g., a relative speedup of 6.39$\times$ was observed for the 16-thread run over 2-thread run, for Rgg\_n\_2\_24\_s0 using the \emph{baseline+VF+Color} implementation.

Another significant contributing factor affecting parallel performance is the time taken to rebuild the graph between consecutive phases.  To analyze this effect, we recorded the breakdown of total run-time by the different phases of the parallel algorithm (described in Section~\ref{secParallelAlgo}).
Figure~\ref{figTimeBreakdown} shows the breakdown - viz. time to rebuild the graph between phases (VF cost is included here), time to perform coloring, and the remaining time attributed to performing the iterations (``clustering''). 
The charts (shown for four representative inputs) explain the discrepancies in scaling among the inputs.
For Rgg\_n\_2\_24\_s0  and MG2, we can see that the time spent in the main clustering iterations dominates, which is desirable from a scaling point of view. 
However, for inputs Europe-osm and NLPKKT240, an increasing portion of time is being spent in the rebuild phase with an increase in the number of cores. 
Given that our current implementation of the rebuild phase has serial bottlenecks (as explained in Section~\ref{secImplementation}), the speedups achieved for these inputs become sub-linear for higher number of cores. 
Figure~\ref{figRebuildAnalysis} confirms these observations about the rebuild phase.
More specifically, for inputs Europe-osm and NLPKKT240, the first phase ends in a low modularity (0.533470 and 0.038107 respectively), which implies that a dominant fraction of the edges remain as inter-community edges. In the graph rebuild phase, each such edge corresponds to two locks (one for each community) affecting parallel performance. In contrast, input MG2 ends with a high modularity score of 0.969587 resulting in an improved performance during the rebuild phase as well.

\begin{figure}[!t]
\centering
\includegraphics[width=2.35in]{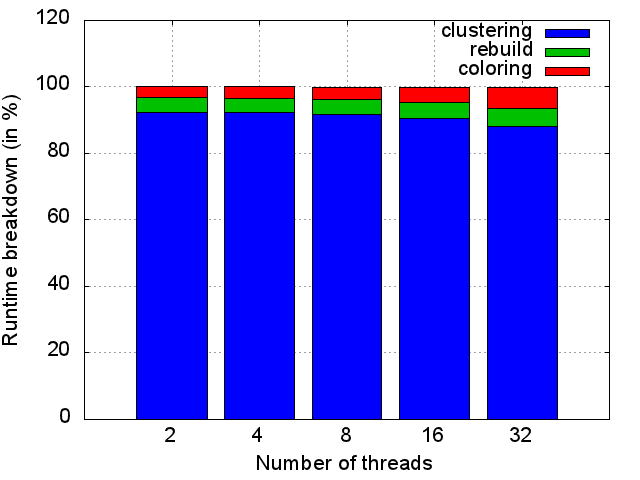}
%\hspace*{0.5in}
\includegraphics[width=2.35in]{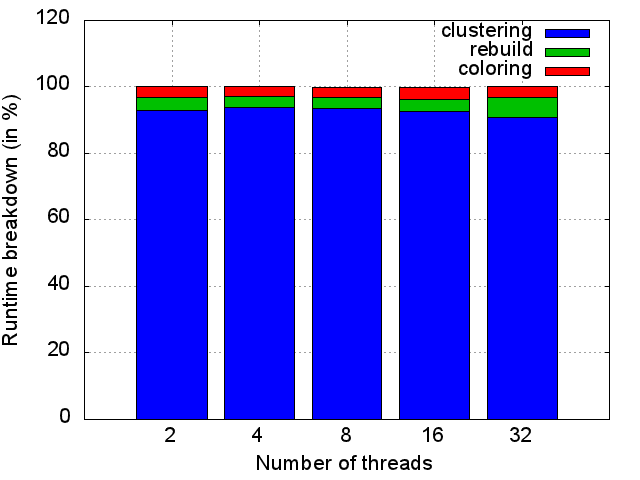}
%\hspace*{0.5in}
\\
{\footnotesize
(a) Input: Rgg\_n\_2\_24\_s0 \hspace*{1.2in} (b) Input: MG2
}
\\
\includegraphics[width=2.35in]{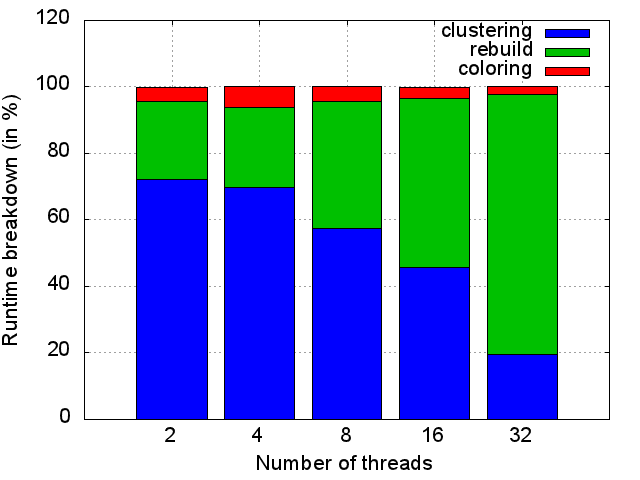}
%\hspace*{0.5in}
\includegraphics[width=2.35in]{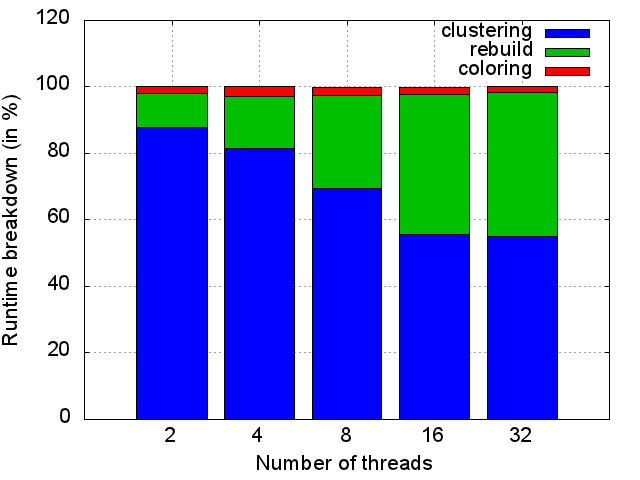}
%\hspace*{0.5in}
\\
{\footnotesize
(c) Input: Europe-osm \hspace*{1.2in} (b) Input: NLPKKT240
}
\caption{ 
Breakdown of the parallel run-times by the different steps of the algorithm - viz. coloring, time to perform the graph transformations between phases, and the time spent in the iterations. The runs correspond to the {\it baseline+VF+Color} implementation.
}
\label{figTimeBreakdown}
\end{figure}

\begin{figure}[!thb]
\centering
{
\includegraphics[width=2.5in]{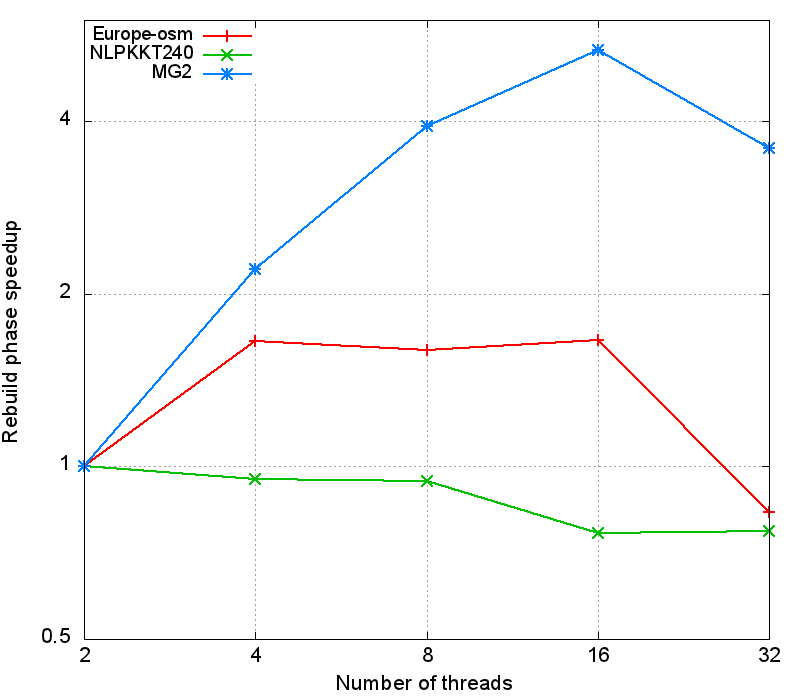}
}
\caption{ 
Chart showing the speedup curves for the graph rebuilding phase of our parallel algorithm.
}
\label{figRebuildAnalysis}
\end{figure}

\begin{table}[!tbh]
 
\caption{
Comparison of the  modularities and run-times achieved by our parallel implementation {\it baseline+VF+Color} (using 8 threads) against the corresponding values achieved by the serial Louvain implementation \cite{louvain_findcommunities_????}.
All runs were performed on the same test platform described under Experimental Setup.
The ``N/A' entries denote cases where the serial Louvain implementation did \emph{not} complete (i.e., crashed). It is to be noted that the serial Louvain implementation is a 32-bit implementation. 
}
\centering
\begin{tabular}{|l|r|r|r|r|r|}
\hline
\multirow{2}{*} {Input }&\multicolumn{2}{c|}{Output modularity}&\multicolumn{3}{c|}{Run-time (in sec)} \\
\cline{2-6}
&Parallel &Serial & Parallel &Serial  & Speedup\\
&&& (8 threads)&& (8 threads) \\
\hline
CNR&0.912608&{\bf 0.912784}&0.8&4.3 & 5.37$\times$\\
\hline
coPapersDBLP&{\bf 0.858088}&0.848702&3.7&7.7 & 2.08$\times$\\
\hline
Channel&{\bf 0.933388}&0.849672&21.2&30.9 & 1.45$\times$\\
\hline
Europe-osm&{\bf 0.994996}&N/A&63.4&N/A & N/A\\
\hline
MG1&{\bf 0.968723}&0.968671&28.8&126.6 & 4.39$\times$\\
\hline
uk-2002&0.989569&{\bf 0.9897}&210.3&335.9 & 1.59$\times$\\
\hline
MG2&0.998397&{\bf 0.998426}&457.8&1313.7 & 2.86$\times$\\
\hline
NLPKKT240&0.934717&{\bf 0.952104}&388.4& 5077.2 & 13.07$\times$\\
\hline
Rgg\_n\_2\_24\_s0&{\bf 0.992698}&0.989637&34.2&111.1 & 3.24$\times$\\
\hline
Soc-LiveJournal&{\bf 0.751404}&0.726785&67.05&182.7 & 2.72$\times$\\
\hline
friendster&{\bf 0.626139}&N/A&2036.8&N/A & N/A\\
\hline
\end{tabular}
\label{tabParSerial}
\end{table}

\subsubsection{Comparison to serial Louvain}
\label{secSerialCompare}

We also comparatively evaluated the performance of our parallel implementations proposed in this paper against the publicly available serial Louvain distribution \cite{louvain_findcommunities_????}.
Figure~\ref{figSpeedup} shows the absolute speedup achieved over the serial implementation for 9 out of the 11 inputs. (For the remaining two inputs, Europe-osm and friendster, the serial implementation failed to run.) 
Table~\ref{tabParSerial} compares the final modularities achieved by both implementations and also the corresponding run-times. 
%The table also shows the time to solution for both methods to enable a direct comparison of both quality and runtime.
For 7 out of the 11 inputs, our parallel implementation delivers higher modularity compared to the serial implementation in shorter time to solution. 
For example, this difference is as much as $>$0.1 for coPapersDBLP and $>0.08$ for Channel. 
Even in 3 out of the 4 cases where the serial implementation delivers higher modularity, the modularities reported by both methods agree up to the first three decimal places. 
Note that the heuristic nature of the algorithm combined with the parallel ordering of vertices which 
could differ from the serial ordering imply that serial and parallel results cannot be guaranteed to be identical. Our results demonstrate that parallelization is at least capable of preserving (if not surpassing) output quality for most of the inputs tested.

As for the run-times, our parallel implementation delivers absolute speedups in the range of  $1.45\times$ to $13.07\times$ using 8 threads.  
Larger speedups were observed using more number of threads, as can be observed from
the absolute speedup chart in Figure~\ref{figSpeedup}. 
A top speedup of 16.51$\times$ was observed for the NLPKKT240 input using 32 cores. 
The two cases where we observe low speedups --- Channel (1.45$\times$)  and uk-2002 (1.59$\times$) --- represent two different cases. 
For the Channel input, observe from Table~\ref{tabInput} that the degree distribution is highly uniform. This could cause vertices to migrate to any one of the neighboring communities and therefore the vertex ordering is expected to have a more pronounced effect on the convergence rate. It is for this reason that the serial implementation, which uses an arbitrary ordering, converges faster albeit with a lower modularity, while our parallel implementation with coloring takes more iterations to converge and does so with a higher modularity.
For uk-2002, the skew in the color set size distribution is the reason behind low speedup (as was explained earlier in the section).

\subsubsection{Performance charts and qualitative evaluation}
\label{secPerfCharts}

Figure~\ref{figPerfCharts} shows the relative performance profiles among the three parallel implementations -- {\it baseline}, {\it baseline+VF}, and {\it baseline+VF+Color} -- along with the serial Louvain implementation for the collection of inputs tested. For plotting these performance charts, we used results from all 9 inputs for which we had results from both serial and parallel implementations. The X-axis represents the factor by which a given scheme fares relative to the best performing scheme for that particular input. The Y-axis represents the fraction of problems (i.e., inputs). The closer a heuristic curve is to the Y-axis the more superior its performance is relative to the other schemes over a wider range of inputs. 
Also, in these performance charts, the order in which inputs appear along each curve is strictly a function of that corresponding heuristic's relative performance to the other schemes --- i.e., the points along a curve are sorted from the corresponding heuristic's best to worst performing inputs. Thus, the charts illustrate the relative performance of each scheme over other schemes for the {\it collection} of 9 inputs tested (as opposed to the individual inputs).

The following observations can be made from the two performance charts. 
The {\it baseline+VF+Color} shows an overall run-time performance advantage over all other schemes.
For instance, consider the run-time curve for {\it baseline+VF+Color} in Figure~\ref{figPerfCharts}b. This implementation outperforms all other heuristics for about 70\% of the problems, about 1.5$\times$ worse compared to a best performing implementation for 20\% of the problems, and 3$\times$ worse than the best for 10 percent of the problems. Similarly, the serial implementation is the slowest ranging from $2\times-5\times$ relative to other best performance schemes. From a modularity standpoint, all parallel heuristics perform comparably to serial method across the input set.

\begin{figure}[!tbh]
\centering
\includegraphics[width=2.50in]{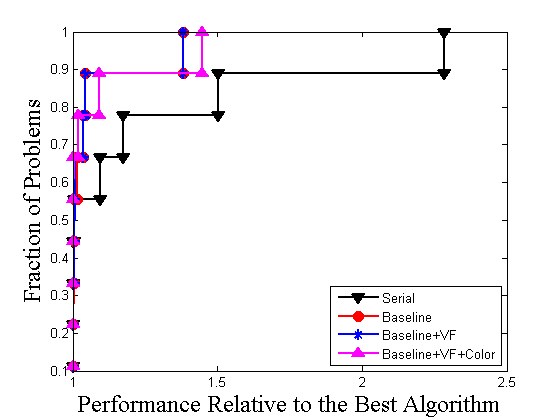}
%\hspace*{0.5in}
\includegraphics[width=2.50in]{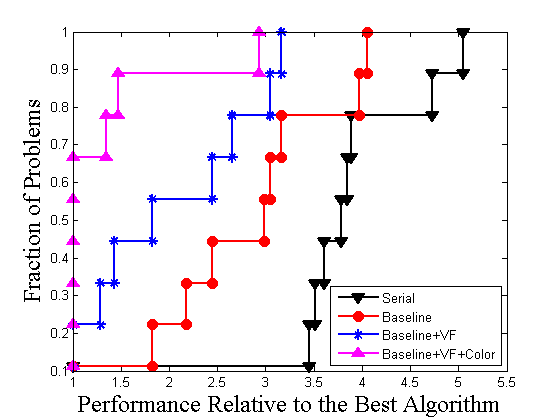}
%\hspace*{0.5in}
\\
{\footnotesize
(a) Modularity profile \hspace*{1.2in} (b) Run-time profile
}
\caption{
Relative profile of performance for three combinations of heuristics: The relative performance of different heuristics and serial implementation for the test problems with respect to the best algorithm for a given problem. Europe-osm and friendster are not included in the comparison because the serial Louvain implementation crashes on those inputs. Final modularity scores are shown in the figure on left (part a), and run-times are shown on the right (part b). Run-time results from 32 thread runs were used to plot curves for the parallel heuristics.
It is to be noted that the longer a heuristic's curve stays near the Y-axis the more superior its performance relative to the other schemes over a wider range of inputs.
}
\label{figPerfCharts}
\end{figure}

{\bf Qualitative comparison:}
In addition to comparing modularities, we also compared the sets of communities by their composition generated by the parallel and serial implementations. The methodology for comparison is as follows. Let $S$ denote the set of communities generated by the serial implementation; and $P$ denote the set of communities generated by one of our parallel implementations --- we used results from the \emph{baseline+VF+Color} for this purpose. Treating the serial output as the ``benchmark'' we compared the parallel output against it as follows. Any vertex pair $(u,v)$ can be categorized into one of the four following bins: 
\begin{itemize}
\item{\bf True Positive (TP):} if $u$ and $v$ belong to the same community in both partitions;
\item{\bf False Positive (FP):} if $u$ and $v$ belong to the same community only in partition $P$;
\item{\bf False Negative (FN):} if $u$ and $v$ belong to the same community only in partitions $S$;
\item{\bf True Negative (TN):} if $u$ and $v$ belong to two different communities in both partitions;
\end{itemize}

Based on the above measures, more qualitative measures, viz. specificity (SP), sensitivity (SE), overlap quality (OQ) and Rand Index, can be calculatated as follows:
$SP = \frac{TP}{TP+FP}$, $SE = \frac{TP}{TP+FN}$, $OQ = \frac{TP}{TP+FP+FN}$, and Rand index $= \frac{TP+TN}{TP+FP+FN+TN}$.

Note that if both results match identically, all these measures will evaluate to 100\%. 
Also note that this comparison takes $\Theta(n^2)$ time because there are $n\choose 2$  pairs.
For this reason, we performed this qualitative comparison only for two of the inputs --- CNR and MG1.

Table~\ref{tabQualitive} shows the results of our comparative study. 
There are two observations that one can make from these results.
First, as can be expected, the partitioning produced by the two methods are different. However, the fact that there is \emph{no} explicit biasing toward false positives or false negatives implies that the cores of communities captured by both methods agree to a large extent --- the OQ values reflect the degree of this agreement. 
Secondly, given that these two partitioning yield nearly identical modularities imply that the vertex pairs consistently grouped by both schemes (i.e., True Positives) contribute to the bulk of the modularity score.

\begin{table}[!htb]
\caption{
Qualitative comparison between the parallel and serial community outputs by their composition. 
}
\centering
\begin{tabular}{|l|r|r|r|r|}
\hline
Input &SP&SE&OQ&Rand index\\

\hline
CNR&83.41\%&89.71\%&76.13\%&99.42\%\\
\hline
MG1&99.60\%&99.83\%&99.43\%&100.00\%\\
\hline
\end{tabular}
\label{tabQualitive}
\end{table}

\subsection{Effect of multiphase coloring}
\label{secPhaseColoring}

\begin{table}[!htb]
\caption{
Comparative results showing the effect of using coloring for only the first phase input vs. for multiple phases of the parallel algorithm.
The multi-phase coloring scheme is same as the {\it baseline+VF+Color} scheme.
All run-times are reported in seconds for runs corresponding to two threads.
}
\centering
\begin{tabular}{|l|r|r|r|r|}
\hline
\multirow{3}{*} {Input}&\multicolumn{2}{c|}{First phase coloring}&\multicolumn{2}{c|}{Multi-phase coloring}\\
\cline{2-5}
& [Min.,Max.] &Run-time  & [Min.,Max.] &Run-time  \\
& Modularity& (\#iter) & Modularity&  (\#iter)\\
\hline
Channel& [0.9344,0.9352] & 103.22 (96) & [0.9304,0.9333] & 52.96 (58) \\
\hline
uk-2002& [0.9895,0.9895] & 670.12 (18) & [0.9894,0.9895] & 748.15 (18) \\
\hline
Europe-osm& [0.9988,0.9988] & 759.94 (306) & [0.9988,0.9989] & 118.97 (38) \\
\hline
MG2& [0.9984,0.9984] & 1422.75 (14) & [0.9984,0.9984] & 1397.90 (12) \\
\hline
\end{tabular}
\label{tabPhaseColoring}
\end{table}

Coloring can be potentially applied to preprocess the input for any phase of the algorithm. However, the time spent coloring is an overhead and a colored graph exposes less parallelism. Therefore, it can be expected that the benefits of coloring, which is to hasten convergence, is expected to diminish as phases progress and the transformed graph becomes smaller. It is for this reason we used a scheme in which coloring is applied until either the number of input vertices reduces below a cutoff (100K for our experiments) or the net modularity gain between phases diminishes below a relatively higher threshold ($10^{-2}$) as described in Section~\ref{secExpSetup}. However, to clearly demonstrate the effect of coloring multiple phases, we devised an alternative implementation in which coloring is applied only to the first phase input. 
The goal was to observe differences in reported modularity and run-times between the two schemes.

Table~\ref{tabPhaseColoring} shows the effect of coloring single phase to multiphase. 
Inputs picked are those for which at least two phases of coloring was applicable.
For the other inputs, the results are identical between single phase and multiphase coloring schemes.
The results demonstrate the benefit of multi-phase coloring as it produces highly comparable modularities over multiple experiments while reducing time-to-solution, for all inputs except uk-2002.

\subsection{Effect of varying the modularity gain threshold}
\label{secThresholdEffect}

\begin{table}[!htb]
\caption{
Table showing the effect of varying the modularity gain threshold.
Two sets of experiments were performed, each running the {\it baseline+VF+Color} implementation, while one using $10^{-2}$ and another $10^{-4}$ as the value for the modularity gain threshold used within the colored phases.
}
\centering
\begin{tabular}{|l|r|r|r|r|}
\hline
\multirow{3}{*} {Input}&\multicolumn{2}{c|}{Threshold = $10^{-4}$}&\multicolumn{2}{c|}{Threshold = $10^{-2}$}\\
\cline{2-5}
& [Min.,Max.] &Run-time  & [Min.,Max.] &Run-time  \\
& Modularity& (\#iter) & Modularity&  (\#iter)\\
\hline
CNR& [0.9125,0.9125] & 5.00 (48) & [0.9125,0.9126] & 1.77 (24)\\
\hline
CoPaperDBLP & [0.8555,0.8577] & 16.17 (27) & [0.8570,0.8580] & 10.64 (23) \\
\hline
Channel & [0.9423,0.9485] & 816.79 (282) & [0.9304,0.9333] & 52.96 (58) \\
\hline
Europe-osm & [0.9989,0.9989] & 250.62 (56) & [0.9947,0.9949] & 125.35 (17) \\
\hline
MG1 & [0.9687,0.9687] & 271.23 (41) & [0.9687,0.9687] & 73.80 (18) \\
\hline
Rgg\_n\_2\_24\_s0 & [0.9926,0.9927] & 227.03 (52) & [0.9926,0.9926] & 118.21 (35) \\
\hline
uk-2002 & [0.9895,0.9896] & 1768.73 (22) &  [0.9894,0.9895] & 748.15 (18) \\
\hline
Nlpktt240 & [0.9426,0.9476] & 3563.41 (147) & [0.9319,0.9347] & 880.94 (78) \\
\hline
MG2& [0.9984,0.9984] & 2652.37 (16) & [0.9983,0.9983] & 1312.44 (7) \\
\hline
\end{tabular}
\label{tabThresholdEffect}
\end{table}

We also studied the effect of varying the modularity gain threshold used within the coloring phases. 
Using a larger value of threshold may prompt phase transitions to happen earlier (and possibly faster convergence) but at the possible expense of the final output modularity. On the other hand, a smaller value could help improve gains within phases but also could prolong phase transitions and eventual completion. 
Two sets of experiments were performed, using values of $10^{-2}$ and $10^{-4}$ for the threshold and the results are summarized in Table~\ref{tabThresholdEffect}.  
As can be observed, the modularities achieved by both schemes are highly comparable, while there is a marked run-time advantage if the threshold is higher.
This study shows that the run-time benefit of using a higher threshold outweighs the qualitative gains of using a lower threshold, at least for the threshold values compared. 

From a modularity standpoint, coloring has a more pronounced effect than the threshold used. The charts in Figure~\ref{figModCurves}a,d,e illustrate this effect --- observe that coloring provides substantail increases in the modularity at the initial phases of the algorithm \emph{before} a finer modularity threshold could take effect in the later phases. 

\section{Related work}
\label{secRelated}

For an extensive review on community detection methods and comparisons, please refer to \cite{fortunato_community_2010,newman_structure_2003}.
Although the notion of community detection is not new, the field took a significant shape with the introduction of the modularity measure to quantify the quality of community outputs by Newman and Girvan in 2004 \cite{newman_finding_2004}.  Newman's pioneering works on discovering community structure from networks also included developing both divisive  \cite{newman_finding_2004,newman_analysis_2004} and agglomerative \cite{clauset_finding_2004} clustering methods. The divisive method use the edge betweenness centrality index to detect bridges between communities but due to the underlying computation involved, it is also very slow ($O(n^3)$ for sparse inputs), limiting its scalability to sparse networks with tens of thousands of vertices.  The other class of algorithms use an agglomerative clustering approach where at any stage a greedy merging is performed between any two communities that provide the maximum modularity gain. This technique was originally introduced by the classical Clauset-Newman-Moore (CNM) algorithm \cite{clauset_finding_2004} and since been adopted/tailored into many other methods (e.g., \cite{wakita_finding_2007}).  With an average time complexity of $O(n \log^2n)$ this approach  have shown better scaling to networks containing $\times 10^5-10^6$ nodes and $\times 10^6-10^7$ edges.  The Louvain method \cite{blondel_fast_2008} can also be thought of as a variant of this agglomerative strategy but with the key differences being that  instead of carrying out the merging at a community-to-community level, the Louvain heuristic allows vertices to independently make decisions from within each community at every time step, and with a flexibility for those decisions to be undone at later iterations. Although input dependent, it has been shown that the Louvain approach is able to produce communities with better modularity scores than the other agglomerative strategies. On the other hand, the cluster hierarchies produced by agglomerative techniques tend to be more meaningful.

In the past few years, there have been several efforts in parallelizing modularity-based community detection. As part of the DIMACS10 clustering challenge, Riedy {\it et al.} presented a highly parallel agglomerative implementation for the CNM algorithm\cite{riedy_scalable_2012,riedy_parallel_2012}. 
Auer and Bisseling \cite{auer_graph_2012} present another way to achieve agglomerative clustering on GPUs using graph coarsening. 
%This method also shows appreciable savings in time to solution.
% --- for instance, an processing an input with 14M vertices and 17M edges is achieved in less than 5 seconds. 
In a recent study, Bhowmick and Srinivasan \cite{bhowmick_template_2013} present a shared memory parallel algorithm for the Louvain method. Their approach is to update the community structures on-the-fly from within each iteration as vertices are evaluated in parallel. This creates a need to introduce critical sections in parts which limits the method's scalability to small synthetic inputs ($\times 10^4$ vertices). The modularities reported also show variability across the processor spectrum. 

There are two parallel efforts to this paper that also describe parallelization of the Louvain algorithm. 
The work by Wickramaarachchi {\it et al.} \cite{wickramaarachchi_fast_2014} targets distributed memory parallelism, with the primary approach being to use a graph partitioner to partition the input graph {\it a priori} and subsequently run the sequential algorithm on each part separately (ignoring the contribution from cross-partition edges) before merging the results through an aggregation process at a master processor. 
In another parallel effort, Staudt and Meyerhenke \cite{staudt_engineering_2013} present an alternative approach called {\it PLM} that uses label propagation to parallelize the Louvain method. 
A comparison of our parallel results with their published results reveals that our parallel implementation {\it baseline+VF+Color} delivers higher modularity than PLM for the inputs both tested --- viz. coPapersDBLP, uk-2002, and Soc-LiveJournal. 
%For instance, for coPapersDBLP, this difference is $\sim 0.06$.
With respect to the run-time performance, a more direct comparison of the two methods on the same platform is required to enable a fair comparison. 
%However, we note that our parallel implementation offers significantly better relative speedup performance than the PLM method. For instance, the PLM implementation offers under $2\times$ speedup on 32 threads relative to its single thread run for uk-2002.

\section{Conclusion}
\label{secConclusion}

In this paper, we introduced effective heuristics for parallelizing an important and widely used community detection method --- the  Louvain method. We attempted to address the dual objectives of maximizing concurrency, and retaining the quality with respect to the serial implementation. To this end, we made two main contributions in this paper. First, we presented 
a detailed discussion of the challenges pertaining to parallelization of the Louvain algorithm 
for community detection, and described effective heuristics to extract parallelism from the algorithm.
Second, we empirically supported the observations with a set of carefully conducted experiments using 11 real-world networks representing a diverse set of application domains. Compared to the serial Louvain implementation \cite{louvain_findcommunities_????}, our parallel implementation is able to produce community outputs with a higher modularity for many of the inputs tested, in comparable number of iterations, while providing real speedups of up to $16\times$ using 32 threads. In addition, our parallel implementation was able to scale linearly up to 16 threads for larger inputs. 

We believe that the mathematical discussion, heuristics, and experimental evidence provided in this paper will benefit a wide range of researchers dealing with increasingly larger data sets and continually weaker serial hardware performance.  
Our future work include: i) extending the experiments to larger-scale inputs with tens of billions of edges and targeting community detection in real-time;  ii) a more thorough comparison of communities produced by the serial and different parallel implementations by delineating differences by composition; iii) investigating the value of the vertex following heuristic in the context of the serial Louvain algorithm and other modularity-based community detection algorithms; and iv) extension of our parallel algorithms to account for alternative modularity definitions (e.g., \cite{traag_narrow_2011}) in order to overcome the known resolution-limit issues of the standard modularity definition used in this paper.


\section*{Acknowledgment}

The authors would like to thank Drs. Emilie Hogan and Daniel Chavarr\'ia for input. The research was in part supported by DOE award DE-SC-0006516, NSF award IIS 0916463, and the Center for Adaptive Super Computing Software Multithreaded Architectures (CASS-MT) at the U.S. Department of Energy Pacific Northwest National
Laboratory (PNNL). PNNL is operated by Battelle Memorial Institute under Contract DE-AC06-76RL01830.
A preliminary version of this paper appeared in \cite{lu_parallel_2014}.


%\section*{References}


\bibliographystyle{elsarticle-num} 
\bibliography{MTAAP14}
%\end{thebibliography}
\end{document}